\documentclass[12pt]{article}
\usepackage[utf8]{inputenc}
\usepackage[T1]{fontenc}
\usepackage[english]{babel}
\usepackage{amsmath,amssymb}
\usepackage{afterpage}
\usepackage{amsthm}
\usepackage{graphicx}
\usepackage[style=ieee,backend=biber,bibencoding=auto,autolang=other]{biblatex}
\usepackage{caption}
\usepackage{subcaption}
\usepackage{enumitem}
\usepackage{amsfonts}
\usepackage{csquotes}
\usepackage{comment}

\DeclareSourcemap{
  \maps{
    \map{
      \step[fieldset=langid, fieldvalue={english}]
      \step[fieldset=hyphenation, fieldvalue={english}]
    }
  }
}

\usepackage{indentfirst}

\makeatletter
\def\tagform@#1{\maketag@@@{(\ignorespaces#1\unskip\@@italiccorr)}}
\renewcommand{\eqref}[1]{\textup{{\normalfont(\ref{#1}}\normalfont)}}
\makeatother

\renewcommand{\phi}{\varphi}

\renewcommand{\leq}{\leqslant}
\renewcommand{\geq}{\geqslant}

\newcommand{\A}{\mathcal{A}}
\newcommand{\B}{\mathcal{B}}
\newcommand{\R}{\mathbb{R}}
\newcommand{\Q}{\mathbb{Q}}
\newcommand{\Z}{\mathbb{Z}}

\newcommand{\m}[1]{\mathrm{#1}\,}

\newcommand{\Cay}{\mathrm{Cay}}
\newcommand{\Aut}{\mathrm{Aut}}
\newcommand{\Vol}{\mathrm{Vol}}
\newcommand{\Conv}{\m{Conv}}
\renewcommand{\S}{S}

\newtheoremstyle{break}
  {\topsep}{\topsep}%
  {\itshape}{}%
  {\bfseries}{}%
  {\newline}{}%

\newtheorem{megathm}{Theorem}
\newtheorem{thm}{Theorem}
\numberwithin{thm}{subsection}
\newtheorem{lem}[thm]{Lemma}
\newtheorem{prp}[thm]{Proposition}
\newtheorem{definition}[thm]{Definition}
\newtheorem{corollary}[thm]{Corollary}

\begin{document}

\title{On~complexity of~multidistance graph recognition in~$\R^1$}
\author{Mikhail Tikhomirov}
\date{}
\maketitle

\begin{abstract}
Let $\A$ be a~set of~positive numbers. A~graph~$G$ is called an~$\A$-embeddable graph in~$\R^d$ if the~vertices of~$G$ can be positioned in~$\R^d$ so that the~distance between endpoints of~any edge is an~element of~$\A$. We consider the computational problem of recognizing~$\A$-embeddable graphs in~$\R^1$ and classify all finite sets~$\A$ by~complexity of~this~problem in~several natural variations.
\end{abstract}

\section{Introduction}

\subsection{Problem statement and motivation}

Let~$\A \subseteq \R_{>0}$ be a~set of~\emph{admissible} distances. For~a~set of~points~$S \subseteq \R^d$ we construct an~\emph{$\A$-distance graph of $S$} as a~graph with~vertices in~$S$ and edges between all pairs of~vertices at~admissible distances. A~generic graph will be called an~\emph{$\A$-distance graph in~$\R^d$} if it is isomorphic to an~$\A$-distance graph of~a~subset of~$\R^d$.

The~notion of~an~$\A$-distance graph is inspired by~classic unit-distance graphs and is, indeed, a~proper generalization since putting~$\A = \{1\}$ yields exactly the~unit-distance graphs. Unit-distance graphs appear in~many classical problems such as Erd\H{o}s' unit distance problem (see~\cite{erdos}), Nelson---Hadwiger problem of~the~chroma\-tic number of~the~plane (see \cite{raigorborsuk}). For a comprehensive survey of these (and many other) discrete geometry problems see~\cite{brass}; a~survey of~results concerning unit-distance graphs can be found in~\cite{raigorcolor,raigorcliques,raigorcoding}. Some isolated properties of~$\A$-distance graphs were studied for~finite sets~$\A$ (see~\cite{raigorcolor,kupav,gorskaya}).

In~literature unit-distance graphs and objects of~similar nature appear under different names, such as~linkages~(\cite{schaefer}), embeddable~(\cite{saxe}), or realizable~(\cite{schaefer}) graphs. Also, the~term ``unit-distance graph'' is sometimes applied to~slightly different objects (e.g., in~\cite{horvat}). Before going further, we find it convenient to~unify the~different notions and extend them to~the~multidistance case.

Let~$G = (V(G), E(G))$ be a~graph. If~$\phi: V(G) \to \R^d$ is such a~map that Euclidean~distance between~$\phi(x)$ and~$\phi(y)$ is an~element of~$\A$ for~all edges~$xy \in E(G)$, then we will say that~$\phi$ is an~\emph{$\A$-embedding of~$G$ in~$\R^d$}.  We will call~$\phi$ an~\emph{injective~$\A$-embedding} if it maps distinct vertices to~distinct points of~$\R^d$ (that is,~$\phi$ is an~injective map). If for~any pair~$xy \not \in E(G)$ we have that the~distance between~$\phi(x)$ and~$\phi(y)$ is not an~element of~$\A$, then~$\phi$ will be called a~\emph{strict~$\A$-embedding}. Naturally, a~graph is~\emph{(strictly) (injectively)~$\A$-embeddable in~$\R^d$} if it admits an~(strict) (injective)~$\A$-embedding in~$\R^d$. Note that strictly injectively~$\A$-embeddable graphs are exactly the~$\A$-distance graphs as~defined above.

We now consider the~computational problem of~recognizing~$\A$-embeddable graphs in~$\R^d$. Throughout the~paper we consider the~distance set~$\A$, the dimension~$d$, as~well as~the~choice of~one of~the~embeddability types (arbitrary, strict, injective, strict and injective) to~be fixed parameters and not~parts of~the~input.

The~complexity of~the~unit-distance ($\A = \{1\}$) case is studied in~\cite{saxe,horvat,schaefer,tikhomirov}: all variations of~the~problem are in~P for~$d = 1$, and are NP-hard for~$d \geq 2$. Another~example of~a~well-studied case is the~case of~$(0, 1]$-distance graphs, usually called \emph{unit ball graphs}. In~the~$d = 1$ case (real line embedding)~$(0, 1]$-distance graphs are~\emph{unit interval graphs}; they are recognizable in~linear time~(\cite{booth,looges}). Re\-cog\-ni\-zing~$(0, 1]$-distance graphs in~the~plane is~NP-hard~(\cite{breu}) and~even hard for~the~existential theory of~reals~(\cite{kang}). An~interesting approach of~\cite{hlinveny} based on~dense lattices in~$\R^d$ allows to~establish NP-hardness~of~$(0, 1]$-embeddability for~$d = 3, 4, 8, 24$.

The~present paper is concerned with the~most ``primitive'' case of~the~$\A$-embeddable graph recognition problem with~$d = 1$ and finite distance sets.
For~each finite set~$\A$ and~each embedding type we classify the~corresponding problem as~belonging to~P or an~NP-complete. Note that since the~set~$\A$ is fixed, all functions that depend only on~$\A$ rather than on~the~input graph are constant in~complexity estimates.

\subsection{Statement of the results}

Let~$\A$ be a~non-empty finite set of~non-zero real numbers. Suppose further that~$\A = -\A$, that is,~$x \in \A$ implies~$-x \in \A$. Let~$G=\langle \A \rangle_+$ be the~additive group generated by~elements of~$\A$. A~graph~$X$ is~$\A$-embeddable in~$\R^1$ if and~only if a~homomorphism of certain type exists between the graphs $X$~and~$\Gamma = \Cay(G, \A)$~--- the~{\it Cayley graph} of~the~group~$G$ with~the~generating set~$\A$. We recall that~by~definition~$\Gamma = (V(\Gamma), E(\Gamma))$ with~$V(\Gamma) = G$ and~$E(\Gamma) = \{\{ x, gx \} \mid x \in G, g \in \A\}$.

The~group~$G$ is~a~free finitely generated abelian group, hence it is isomorphic to~$\Z^k$ for~an~integer~$k \geq 1$. In~the~sequel we identify each element of~$G$ with~the~element of~$\Z^k$ being its image under~a~certain canonically chosen group isomorphism.

The~following theorem provides a~complete classification of~finite sets~$\A$ depending on~the~complexity of~$\A$-embeddability checking in~$\R^1$.

\begin{megathm} \label{summary}

\begin{enumerate}[label=(\alph*)]

\item The~problem of~$\A$-embeddability checking in~$\R^1$ is in~P if the~graph~$\Gamma$ is~bipartite, otherwise the~problem is~NP-complete.

\item The~problem of~strict~and/or injective~$\A$-embeddability checking in~$\R^1$ is in~P if~$\langle \A \rangle_+ \sim \Z$, otherwise the~problem is~NP-complete.

\end{enumerate}

\end{megathm}

Note that~$\langle \A \rangle_+ \sim \Z$ if~and~only if all pairwise quotients of~elements of~$\A$ are rational, or, equivalently,~$\A \subset \alpha \Z$ for~a~real~$\alpha \neq 0$.

The~(a) part of~Theorem~\ref{summary} is~an~immediate corollary of the~result~\cite{macgillivray} on the~com\-ple\-xi\-ty of~$H$-coloring for~infinite graphs~$H$ of~bounded degree. It is possible to obtain a more explicit condition in terms of elements of $\A$:

\begin{prp}\label{gammabipart}

If $\A \subset \Z^k$ is a symmetrical generating set of $\Z^k$, then $\Gamma = \Cay(\Z^k, \A)$ is bipartite iff there is a subset $I \subseteq \{1, \ldots, k\}$ such that for each $x = (x_1, \ldots, x_k) \in \A$ we have that $\sum_{i \in I} x_i$ is odd.

\end{prp}

\begin{proof}

If there is a set $I$ that satisfies the premise, then $\Gamma$ is bipartite with parts~$\{G_0, G_1\}$ defined by $G_j = \{x \in \Z^k \vert \sum_{i \in I} x_i \equiv j~(\bmod~2)\}$. Conversely, let $\Gamma$ be bipartite with parts $\{G_0, G_1\}$, i.e. $G_0, G_1 \subseteq \Z^k$, $G_0 \cup G_1 = \Z^k$, $G_0 \cap G_1 = \varnothing$, and for any edge $xy$ of $\Gamma$ neither $G_0$ nor $G_1$ contains both $x$ and $y$. Note that if $\Gamma$ is bipartite, then there is only one correct partition. Without loss of generality, assume that $0 \in G_0$. Consider a basis element $e_i \in \Z^k$, with all coordinates except $i$-th equal to 0, and $i$-th coordinate equal to 1. If $e_i \in G_0$, then by vertex transitivity we must have that for any $x \in \Z^k$ the elements $x$ and $x + e_i$ belong to the same part. If $e_i \in B$, then $x$ and $x + e_i$ must belong to different parts for any $x$. Define $I = \{i \in \{1, \ldots, k\} \vert e_i \in G_1\}$. For any element $x \in \A$ we must have $x \in G_1$ since $0x$ is an edge of $\Gamma$, hence $\sum_{i \in I} x_i$ is odd, which concludes the proof.

\end{proof}

The~case~$\langle \A \rangle_+ \sim \Z$ of~the~(b)~part follows from~the~result~\cite{matousek} on~the~time-polynomial~solution of~{\bf SUB\-GRAPH-ISOMORPHISM} for~graphs of~bounded tree\-width; the~details are~given in~Sec\-tion~\ref{z1} (a~discussion of~treewidth and~a~survey of~relevant algorithmic results can be~found in~\cite{bodlaendertourist}). The~bulk of~the~present paper is dedicated to~proving NP-comp\-lete\-ness of~strict~and/or~injective $\A$-embeddability checking in~$\R^1$ in~the~case~$\langle \A \rangle_+ \sim \Z^k$ with~$k \geq 2$. Let us outline the scheme of the proof.

In~Section~\ref{balls} we study automorphisms of~the~Cayley graph~$\Gamma$ and~its finite subgraphs. The~main result of~the~section is Theorem~\ref{finiteball} that asserts existence of~finite subgraphs of~$\Gamma$ such~that each of~their automorphisms acts linearly on~elements of~$\Z^k$ and can be extended uniquely to~an~automorphism of~the~full graph~$\Gamma$. The~existence of~``$\Gamma$-rigid'' subgraphs provided by~Theorem~\ref{finiteball} allows us to~avoid most of~the~complications arising from~the~``graphical'' nature of~$\A$-embeddability and lead the~discussion of~the~subsequent constructions in~geometric terms.

In~order to~establish NP-completeness, we implement the~``logic engine'' setup (see Section~\ref{enginedesc}) to~reduce from~the~NP-complete {\bf NAE-3-SAT} problem to~strict and/or injective~$\A$-embeddability checking via an intermediate problem {\bf LOGIC-ENGINE}. In~Section~\ref{zk} we describe the~reduction from~logic engine realizability to~each case of~strict~and/or~injective~$\A$-embeddability checking in $\R^1$ using two different logic engine implementations for~the~cases~$k = 2$ and~$k > 2$.

Note that all embeddability problems in~$\R^1$ belong to~NP since they are equivalent to~the~$\Gamma$-coloring problem  (with~possible additional constraints) which admits polynomial certificate, namely, integer coordinates of~corresponding elements of~$\Z^k$.

\subsection{The~$G \sim \Z$~case, strict~and/or injective~$\A$-embeddability} \label{z1}

Since~$\langle \A \rangle = \Z$, the elements~of~$\A$ become mutually coprime integers under~a ca\-no\-ni\-cal group iso\-morphism of~$G$ and~$\Z$. Let us put~$D = \max \{ |x|: x \in \A \}$.

First consider the~injective (possibly non-strict) embeddability case.

\begin{lem}

Let $H$ be a~finite subgraph of~$\Gamma = \Cay(\Z, \A)$. Then the treewidth (and, moreover, the pathwidth) of~$H$ is at~most~$D$.

\end{lem}

\begin{proof}

Note that adding edges to~a~graph does~not decrease its path- or treewidth. Without loss of~generality, suppose that~the~subgraph~$H = (V, E)$ is~induced by~a~vertex set~$V = \{0, \ldots, M\}$. If $M < D$, then take the~trivial path decomposition with~a~single vertex $V$. This decomposition has width~$M < D$, hence~the~claim holds.

Now suppose that~$M \geq D$. We will build a~path decomposition~$P$ of~the~graph~$H$ with width~$D$. Take subsets~$A_x = \{x, \ldots, x + D\}$ for~all integer~$x$ from~0 to~$M - D$ as~vertices of~$P$. The~edges of~$P$ will connect subsets that are different in~a~single element.

Let us ensure that~$P$ is indeed a path decomposition of~$H$. Clearly~$P$ is~a~path, and~for~every vertex~$x \in V$ the~vertices~of~$P$ containing~$x$ form~a~subpath. Suppose that~$xy \in E$ and~$x < y$. Then~$y \leq x + D$ and~$x, y \in A_x$, hence both endpoints of~any edge of~$H$ are covered by~a~vertex~of~$P$. Thus all requirements of~a~path decomposition are met. Finally, it can easily be seen that the~width~of~$P$ is~$D$.

\end{proof}

Suppose that~a~connected graph~$X = (V(X), E(X))$ is the~input to~the~$\A$-em\-bed\-da\-bi\-li\-ty checking problem, and~$Y$ is~a~subgraph~of~$\Gamma$ induced~by~the~vertex set~\{$0$, \ldots, $|V(X)| \cdot D$\}. Clearly,~the~graph~$X$ is~(strictly/non-strictly) injectively~$\A$-embeddable in~$\R^1$ if and~only if~$X$ is isomorphic to~an~(induced/non-induced) subgraph of~$Y$.

The~following result is due to~\cite{matousek}: suppose that the~maximal degree of~a~connected graph~$X$ is bounded by~a~constant~$\Delta$, and a~graph~$Y$ has treewidth bounded by~a~constant~$k$, then finding an~(induced/non-induced) subgraph~of~$Y$ that is isomorphic to~$X$ can be done in~$O(|V(X)|^{k + 1} |V(Y)|)$ time.
The~maximal degree of~the~graph~$Y$ is at~most~$2D$, thus we can assume that the~maximal degree of~$X$ is at~most~$2D$ as~well (otherwise~$X$ can not be isomorphic to~a~subgraph of~$Y$). Further, by~the~previous lemma the~treewidth of~$Y$ is at~most~$D$. Thus, by~using the~algorithm of~\cite{matousek}, we obtain an~algorithm for~checking injective (strict/non-strict)~$\A$-embeddability in~$O(|V(X)|^{D + 2})$ time.

Finally, consider the case of~strict non-injective~$\A$-embeddability. Let~$N(v)$ denote the~set of~neighbours of~a~vertex~$v$ in~the~graph~$X$. We will say that vertices~$v, u \in V(X)$ are~\emph{equivalent} if~$N(v) = N(u)$, and will write $v \sim u$.

\begin{prp} \label{notinj}

Suppose that~$V'$ is~a~subset of~$V(X)$ that contains a~single vertex from~each equivalence class of~$V(X)$, and $X'$~is~the~subgraph of~$X$ induced by~$V'$. Then the~graph~$X$ is strictly~$\A$-embeddable in~$\R^1$ if and~only if~$X'$ is strictly injectively~$\A$-embeddable in~$\R^1$.

\end{prp}

\begin{proof}

Suppose that~$\phi$ is a~strict~$\A$-embedding of~$X$ in~$\R^1$. If~$\phi(v) = \phi(u)$, then by~strictness of~$\phi$ for every other vertex~$w \in V(X)$ the~edges~$vw$ and~$uw$ are either both inside or both outside of~$E(X)$, and we must have~$v \sim u$. Since no two distinct vertices of $V'$ are equivalent, then the~restriction~$\phi|_{V'}$ is a~suitable strict injective~$\A$-embedding of~the~graph~$X'$.

Conversely, consider a~strict injective~$\A$-embedding~$\phi'$ of~the~graph~$X'$ in~$\R^1$. Define an~embedding~$\phi: V(X) \to V(\Gamma)$ by~$\phi(v) = \phi'(R(v))$, where~$R(v)$ is the~only vertex of $V'$ that satisfies $v \sim R(v)$. If~$vu \in E(X)$, then we must have~$R(v)R(u) \in E(X)$. But~$\phi$ is strict, hence~$\phi(x) - \phi(y) = \phi(R(x)) - \phi(R(y)) \in \A$, thus~$\phi$ is a~strict~$\A$-embedding.

\end{proof}

The~graph~$X'$ can be easily constructed by~$X$ in~polynomial time, and~strict injective~$\A$-embeddability of~$X'$ can be checked in~polynomial time in~the~$k = 1$ case. Thus the~first half of~the~(b) part of~Theorem~\ref{summary} is proven.

\section{Balls in $\Gamma$ and their automorphisms} \label{balls}

\subsection{Balls and embeddings}

Recall that~$G \sim \Z^k$ is~a~free finitely generated abelian group,~$\A$ is a~finite generating set of~$G$, and~$\Gamma = \Cay(G, \A)$.

Let~$x, y \in G$. Since for~each~$d \in G$ the~translation~$x \rightarrow x + d$ is an~automorphism of~$\Gamma$, the~length of~the~shortest path between vertices~$x$ and~$y$ in~the~graph~$\Gamma$ depends only on~$x - y$; let~$\rho^{\A}(x - y)$ denote this length ($\rho^{\A}$~is~the~same as~the~\emph{word metric} of~the~group~$G$ with~the~generating set~$\A$). Also let~$\omega^{\A}(x - y)$ denote the~number of~shortest paths between~$x$ and~$y$ in~$\Gamma$. In the sequel we will omit the~upper index~$\A$  when the~set~$\A$ is clear from~the~con\-text.

If~$x \in G$ and $V \subset G$, then we will write~$V + x$ for~a~copy of~$V$ translated by~$x$ element-wise. Similarly, if $\lambda$ is a~subgraph of~$\Gamma$, we will write~$\lambda + x$ for~a~subgraph obtained from~$\lambda$ by~shifting all vertices and~endpoints of~edges by~$x$.

\begin{definition}

Suppose that~$r$ is~a~non-negative integer. Let~$B^{\A}_r$ denote~the~subgraph of~$\Gamma$ induced by~the~vertex set~$V(B^{\A}_r) = \{x \in G: \rho^{\A}(x) \leq r\}$. We will call the~graph~$B^{\A}_r$ the~\emph{ball of~radius~$r$}.

\end{definition}

\begin{prp} \label{balliso}

Suppose that~$\lambda$ is a~subgraph of~$\Gamma$, and~$\phi: B_r \to \lambda$ is a~graph isomorphism. Then~$\lambda = B_r + \phi(0)$.

\end{prp}

\begin{proof}

Since for~any path~$v_0$, \ldots, $v_s$ in~the~graph~$B_r$ there is a~path~$\phi(v_0)$,~\ldots, $\phi(v_s)$ in~the~graph~$\lambda$, then for~any vertex~$x \in V(B_r)$ we have
\[\rho(\phi(x) - \phi(0)) \leq \rho(x) \leq r,\]
hence~$V(\lambda) \subseteq V(B_r + \phi(0))$. But
\[|V(\lambda)| = |V(B_r)| = |V(B_r + \phi(0))|,\] thus we have~$V(\lambda) = V(B_r + \phi(0))$. Edges of~$B_r$ and~$\lambda$ are in~one-to-one cor\-res\-pon\-dence, thus
\[|E(B_r)| = |E(\lambda)| \leq |E(B_r + \phi(0))| = |E(B_r)|,\] and~$E(\lambda) = E(B_r + \phi(0))$.

\end{proof}

We are interested in~possible \emph{embeddings} of~the~ball of~radius~$r$ in~the~graph~$\Gamma$, that is, graph isomorphisms~$\phi: B_r \to B_r + \phi(0)$. Since~$\Gamma$ is vertex-transitive, it suffices to~consider the~group~$\Aut(B_r)$, because every embedding~$\phi$ is composed of an~auto\-mor\-phism of~$B_r$ and a~translation by~$\phi(0)$.

Consider~the~group~$\Aut_0(\Gamma)$ of~automorphisms of~$\Gamma$ that stabilize the~origin. It~follows from the results of~\cite{ryabchenko} that each automorphism~$\phi \in \Aut_0(\Gamma)$ is additive (that~is, satisfies~$\phi(x + y) = \phi(x) + \phi(y)$ for~all~$x, y \in G$), hence it is unambigiosly determined by~the~values~$\phi(a)$ on~all~$a \in \A$.

\emph{Linearity} is a~stronger property of~an~automorphism. We will say that $\phi \in \Aut_0(\Gamma)$ is linear if there exists a~non-degenerate linear map~$T: \R^k \to \R^k$ such that~$\phi(x) = T(x)$ for~all~$x \in G$. Clearly, linearity implies additivity.

From each~$\phi \in \Aut_0(\Gamma)$ we can construct an~element of~$\Aut(B_r)$, namely, the~re\-stric\-tion~$\phi|_{V(B_r)}$. Since an~element of~$\Aut_0(\Gamma)$ is induced by its values on~$\A$, we have that distinct elements of~$\Aut_0(\Gamma)$ have distinct restrictions on~$B_r$ (if~$r$ is positive). It should be noted, however, that in many cases the graph $B_r$ admits different kinds of automorphisms: for example, if $\A$ is a standard basis of $\Z^k$ (after adjoining the inverse elements) and $r = 1$, then the group $\Aut_0(\Gamma)$ is isomorphic to $\S_k \times \Z_2^k$ since each automorphism can freely exchange the axes and/or flip their directions. At the same time, we have that~$\Aut(B_r)$ is isomorphic to~$\S_{2k}$ since the~graph~$B_r$ is isomorphic to~$K_{1, 2k}$. Moreover, for~each integer~$r > 0$ we can construct a~set~$\A$ such that~$\Aut(B_r) \not \sim \Aut_0(\Gamma)$: if we take, for example,\[\A = \{\pm(2r + 1, 0), \pm(0, 2r + 2), \pm(1, 1)\},\] then the~graph~$B_r$ is~isomorphic to~a~ball of~radius~$r$ in~the~lattice~$\Z^3$ with~standard basis, hence~$\Aut(B_r) \sim \S_3 \times \Z_2^3$, while~$\Aut_0(\Gamma) \sim \Z_2$ and contains only the trivial automorphism and the central symmetry.

Theorem \ref{finiteball} shows that for~sufficiently large radius~$r$ the~group~$\Aut(B_r)$ precisely captures the~structure of~$\Aut_0(\Gamma)$, or, in~other words, each automorphism of~$B_r$ can be extended to~an~automorphism of~$\Gamma$ (in~this case, the~extension is~unique). Moreover, all automorphisms of~$B_r$ and~$\Gamma$ turn out to~be linear.

\begin{megathm} \label{finiteball}

Suppose that~$G$ is~a~free finitely generated abelian group, and~$\A \subset G$ is an~arbitrary finite symmetrical generating set of~$G$ that doesn't contain~0. Then there exists an~integer~$R(\A)$ such that for~every integer~$r \geq R(\A)$ each automorphism of~the~graph~$B_r$ is linear and has a unique extension to~an~element of~$\Aut_0(\Cay(G, \A))$.

\end{megathm}

Let us outline the~proof of~Theorem~\ref{finiteball}. First, we establish that each automorphism of~a~sufficiently large ball stabilizes the~set of~vertices of~convex hull of~$\A$ (these vertices are called the~{\it primary} elements of~$\A$). We will refer to~the~corresponding permutation of~primary elements of~$\A$ as an~{\it orientation} of~the~automorphism. We also show that orientation is well-defined for any ``bundle'' of~balls, that is, the induced union of~balls with~a~connected set of~centers. We will use this to~show that each automorphism of~a~large enough ball is linear on~the~partial lattice induced by~primary elements.

On~the~other hand, we will show that the~restriction of~any automorphism to~the~lattice induced by non-primary ({\it secondary}) elements is an~automorphism of~the~ball in the Cayley graph of the additive group induced by the secondary elements. By~induction on~the~size of~$\A$, this implies linearity of~this~automorphism if the~ball is large enough. The final part of the proof is ``gluing'' of~the~two linear automorphisms and extending them to all vertices of the ball that do not belong to any of the lattices. Finally, we exclude each automorphism of~$B_r$ that does not allow extension to~an~automorphism of~$\Gamma$ simply by increasing $r$ since the~set of~permutations of~$\A$, and therefore, of~possible automorphisms, is finite.

\subsection{Properties of ball embeddings in $\Gamma$}

\begin{prp} \label{rhopres}

For~each graph automorphism~$\phi: B_r \to B_r + \phi(0)$ and each vertex~$x \in V(B_r)$ the~equalities~$\rho(\phi(x) - \phi(0)) = \rho(x)$ and~$\omega(\phi(x) - \phi(0)) = \omega(x)$~hold.

\end{prp}

\begin{proof}

For~each ball~subgraph $B_r + z$ any shortest path in~$\Gamma$ between the~center~$z$ and any vertex of~$B_r + z$ contains only edges of~$B_r + z$, hence there is a~bijection between shortest paths in~$\Gamma$ from~$0$ to~$x$ and from~$\phi(0)$ to~$\phi(x)$.

\end{proof}
\begin{definition}

We will say that~$x \in \A$ is a~\emph{primary} element if for~any integer~$t \geq 0$ we have~$\rho(tx) = t$ and~$\omega(tx) = 1$. Any other element of~$\A$ will be a~\emph{secondary} element.

\end{definition}

Let~$\A'$ denote the~set of~all primary elements of~$\A$. It is clear from~the~symmetry argument that~$\A' = -\A'$.

By definition, there exists an~integer~$D$ such that for~all secondary~$x \in \A$ we have either~$\rho(Dx) \neq D$ or~$\omega(Dx) \neq 1$. Aside from~characterization of~primary elements of~$\A$, the~next proposition contains a~constructive way of~choosing a~suitable value of~$D$ in~the~second part of~the~proof.

If~$V \subset \R^n$, then we will write~$\Conv V$ for~the~convex hull of~the~set~$V$.

\begin{prp} \label{mainconv}

An~element~$x \in \A$ is primary if and only if~$x$ is a~vertex of~$\Conv \A$.

\end{prp}

\begin{proof}

Let~$x$ be a~vertex of~$\Conv \A$, then there is a~linear function~$l: \R^k \to \R$ such that~$l(x) > l(y)$ for~all~$y \in \A$ that are different from~$x$. It follows that~$l(x) > l(-x) = -l(x)$, and~$l(x) > 0$. Suppose that~$d_1, \ldots, d_s \in \A$ is a~sequence of~edge transitions in $\Gamma$ leading from~0 to~$tx$, that is,~$d_1 + \ldots + d_s = tx$. Then we have~\[tl(x) = l(tx) = l(d_1) + \ldots + l(d_s) \leq sl(x),\] hence~$s \geq t$ and~$\rho(tx) \geq t$. If~$s = t$ and not~all elements~$d_1$,~\ldots,~$d_t$ are equal to~$x$, then~$l(d_1) + \ldots + l(d_t) < tl(x)$, a~contradiction. Thus the~shortest path in $\Gamma$ between~0 to~$tx$ has length~$t$
and is unique, hence $x$ is indeed a~primary element.

Now suppose that~$x$ is not a~vertex of~$\Conv \A$. Note that the~origin of~$\R^k$ lies inside of~$\Conv \A$ since~$\A = -\A$. Then~$x$ belongs to~the~convex hull of~the~origin and a~certain set~$y_1, \ldots, y_s \in \A$ (with all~${y_i \neq x}$) over the field~$\Q$. Hence we have~$x = \sum \alpha_i y_i$ for some non-negative rational numbers~$\alpha_i$ such that $\sum \alpha_i \leq 1$. After~multiplying by~a~common denominator of~the~numbers~$\alpha_i$ we obtain the~equality~$Kx = \sum A_i y_i$, where~$K$ and all~$A_i$ are non-negative integers, and~$0 < \sum A_i \leq K$. From~this equality we obtain that there is a~path in~$\Gamma$ between~0 and~$Kx$ of~length at~most~$K$ that contains transitions other than~$x$, thus~$\phi(Kx) < K$ or~$\omega(Kx) \neq 1$, and~$x$ is not a~primary element.

\end{proof}

From~the~second part of~the~proof we can also obtain the~following

\begin{prp} \label{mainshortcut}

For~each element~$x \in \A$ there exists an~integer~$K_x$ such that~$K_x x$ can be represented as a~sum of~at~most~$K_x$ primary elements (with repetitions allowed).

\end{prp}

The~next proposition can be informally stated as~follows: in~any embedding of~a~large enough ball in~$\Gamma$ the~``rays'' (i.\,e., the~unique shortest paths) that correspond to~the~primary directions are <<rigid>> and can only be permuted among~themselves, while opposite rays stay opposite and form a~``straight line''.

\begin{prp} \label{ballmain}

Suppose that~$r \geq \max(2, D)$ and~$\phi: B_r \to B_r + \phi(0)$ is a~graph isomorphism. Then for each~$x \in \A'$ there exists~$x'\in \A'$ such that~$\phi(rx) - \phi(0) = rx'$. Moreover,~$\phi(tx) - \phi(0) = tx'$ for~all integer~$t \in [-r, r]$.

\end{prp}

\begin{proof}

Proposition \ref{rhopres} implies that\[\rho(\phi(rx) - \phi(0)) = \rho(rx) = r\]and\[\omega(\phi(rx) - \phi(0)) = \omega(rx) = 1.\] However, if $\phi(rx) - \phi(0) \neq rx'$ for all $x' \in \A$, then we must have $\omega(\phi(rx) - \phi(0)) \neq 1$. Indeed, an arbitrarily chosen shortest path from $\phi(0)$ to $\phi(rx)$ must contain different transitions, thus by permuting them we obtain a different shortest path. Therefore we have~$\phi(rx) - \phi(0) = rx'$ for~a~certain~$x' \in \A$. Moreover, $r \geq D$, thus $x'$ is a primary element by the choice of $D$.

Next, let us establish the~second part of~the~statement for~$t = -r$, that is, \[\phi(-rx) - \phi(0) = -rx'.\] Suppose that~$\phi(-rx) - \phi(0) = rx''$ with~$x'' \neq -x'$. Note that there is a~unique path of~length~$2r$ between~$rx$ and~$-rx$ in~the~graph~$B_r$ since~$x$ is a~primary element. At~the~same time, consider the~following path of~length~$2r$ from~$\phi(rx) = \phi(0) + rx'$ to~$\phi(-rx) = \phi(0) + rx''$ that passes through~$\phi(0)$: $r$ transitions by~$-x'$ followed by~$r$ transitions by~$x''$. Let us exchange~$r$-th and~$(r + 1)$-th step of~this path, the~resulting path will pass through~$\phi(0) + x'' - x'$ instead of~$\phi(0)$. Since~$r \geq 2$ and~$\rho(x'' - x') \leq 2$, the new path lies completely inside of~$B_r + \phi(0)$, and thus we have at~least two different paths of~length~$2r$ between $\phi(rx)$ and $\phi(-rx)$ in $B_r + \phi(0)$. It follows that the~graph isomorphism~$\phi$ does not preserve the~number of~paths of~length~$2r$ between a~pair of~vertices, which is a contradiction. Hence we have $\phi(-rx) - \phi(0) = -rx'$.

Finally, let us prove the~second part for~all other values of~$t$. First, let~$t \geq 0$. The~vertex~$\phi(tx)$ must belong to~the~only shortest path between~$\phi(0)$ and~$\phi(rx)$. But this path consists only of~transitions by~$x'$, hence~$\phi(tx) - \phi(0) = tx'$. The $t \leq 0$ case is handled similarly by considering the shortest path between $\phi(0)$ and $\phi(-rx)$.

\end{proof}

Let us define the~orientation of~a~ball embedding as a~way of~permuting the~primary elements.

\begin{definition}\label{defchi}

Suppose that~$r \geq \max(2, D)$ and~$\phi: B_r + z \to B_r + \phi(z)$ is a~graph isomorphism. Let~$\chi^{\phi}$ denote the~\emph{orientation} of~the~isomorphism~$\phi$ as~a~function that maps each~$x \in \A'$ to~$\phi(x + z) - \phi(z)$. If~$\phi: X \to Y$ is an~isomorphism between two subgraphs of~$\Gamma$, and~$B_r + z$ is a~subgraph of~$X$, let us write~$\chi^{\phi}_{B_r + z}$ for~the~orientation of~the~restriction of~$\phi$ to~$B_r + z$.

\end{definition}

Proposition \ref{ballmain} immediately implies

\begin{corollary}

$\chi^{\phi}$ is a permutation of $\A'$ that satisfies $\chi^{\phi}(x) = -\chi^{\phi}(-x)$.

\end{corollary}

\begin{definition}
For~$x \in G$ define the~\emph{norm~$||x||$} as an~$l_{\infty}$-norm of~the~corresponding element of~$\Z^k$ (i.\,e., the~largest absolute value of~coordinates). For~a~subset~$V \subset G$ put~$||V|| = \max_{x \in V} ||x||$.
\end{definition}

The~next proposition states that if two large enough balls (possibly having common vertices) are subgraphs of~an~induced subgraph of~$\Gamma$, and the~centers of~the~balls are adjacent, then their orientations must coincide in~any embedding of~the subgraph.

\begin{prp} \label{adjballs}

Suppose that~$r \geq \max(2, D, 2||\A|| + 1)$. Suppose further that~${z \in \A}$, that~$B_r$ and~$B'_r = B_r + z$ are balls of~radius~$r$ with~centers at~vertices~$0$ and~$z$ respectively, and that~$C$ is~the~subgraph of~$\Gamma$ induced by~vertices of~$B_r$ and~$B'_r$. Finally, suppose that~$\lambda$ is a~subgraph of~$\Gamma$, and~$\phi: C \to \lambda$ is a~graph isomorphism of~$C$ and $\lambda$. Then the~restrictions of~$\phi$ to~$B_r$ and~$B'_r$ have equal orientations.

\end{prp}

\begin{proof}

Let us assume that $\chi^{\phi}_{B_r}(x) \neq \chi^{\phi}_{B'_r}(x)$ for some $x \in \A'$, that~is, \[\phi(x) = \phi(0) + x',~\phi(x + z) = \phi(z) + x'',~x' \neq x''.\] Proposition~\ref{ballmain} implies that\[\phi(rx) = \phi(0) + rx',~\phi(rx + z) = \phi(z) + rx''.\]Since~$\{rx, rx + z\}$ is an~edge of~$C$, we have the inequality \[||\phi(rx + z) - \phi(rx)|| \leq ||\A||.\] But \[||\phi(rx + z) - \phi(rx)|| = ||r(x'' - x') + \phi(z) - \phi(0)|| \geq \]\[ r||x'' - x'|| -||\phi(z) - \phi(0)|| \geq r - ||\A|| > ||\A||,\] which is a contradiction.

\end{proof}

\begin{corollary} \label{adjballs2}

In~the~assumptions of~Proposition~\ref{adjballs}, if we additionally have~${z \in \A'}$, then~$\phi(tz) - \phi(0) = t(\phi(z) - \phi(0))$ for~all integer~$t \in [-r, r + 1]$.

\end{corollary}

Informally this corollary can be restated as follows: if one of~the~balls has its center on~the~``line'' that corresponds to a~primary direction of~another ball, then in~any embedding the~``lines'' that correspond to~this~direction in~both balls must be aligned.

\begin{proof}
Only the~case~$t = r + 1$ is uncovered by~the~Proposition~\ref{ballmain}. Applying this proposition to~$B'_r$ we have~$\phi((r + 1)z) = \phi(z) + r\chi^{\phi}_{B'_r}(z)$. But \[\chi^{\phi}_{B'_r}(z) = -\chi^{\phi}_{B'_r}(-z) = -(\phi(0) - \phi(z)) = \phi(z) - \phi(0).\] After~substitution and transfer of~$-\phi(0)$ to~the~left-hand side we obtain the~desired equality.

\end{proof}

\begin{corollary} \label{connballs}

Suppose that~$r \geq \max(2, D, 2||\A|| + 1)$. Further, let~$M \subset G$ be a~connected subset of~vertices of~$\Gamma$, and~$C$ is the~subgraph of~$\Gamma$ induced by~the~set~\[\displaystyle\bigcup_{x \in M} V(B_r + x).\] Finally, suppose that~$\lambda$ is a~subgraph of~$\Gamma$, and~$\phi: C \to \lambda$ is a graph isomorphism of~$C$ and~$\lambda$. Then restrictions of $\phi$ to $B_r + x$ have equal orientation for all $x \in M$.

\end{corollary}

We will write~$\chi^{\phi}$ for~the~orientation of~restriction of~$\phi$ to~any ball when Corollary~\ref{connballs} is applicable.

\subsection{Lattices induced by~primary and secondary elements}

Before proceeding, let us point out a~simple fact.

\begin{prp} \label{containment}

Suppose that~$x$ is a~vertex of~$B_r$, and~$d$ is a~non-negative integer that satisfies~$d + \rho(x) \leq r$. Then~$B_d + x$ is a~subgraph of~$B_r$.

\end{prp}

\begin{proof}

Each vertex~$v$ of~the~graph~$B_d + x$ can be represented as~$v = z + x$, where~$z \in B_d$, hence~$\rho(v) \leq \rho(z) + \rho(x) \leq d + (r - d) = r$, which implies~$v \in V(B_r)$. Further, both subgraphs are induced, thus~$E(B_d + x) \subseteq E(B_r)$.

\end{proof}

Now, let us show that each automorphism of~a~large enough ball is additive on~sums of~primary elements of~$\A$.

\begin{lem} \label{mainsum}
Suppose that~$\A' = \{a_1, \ldots, a_s\}$. Denote~$r_0 = \max(2, D, 2||\A|| + 1)$, and put~$R = sr_0$. Suppose that~$r \geq R$ is an~integer,~$\phi$ is an~automorphism of~$B_r$, and~$\alpha_1$,~\ldots,~$\alpha_s$ are non-negative integers which sum does not exceed $r$. Then the~following equality holds:
\[\phi(\alpha_1 a_1 + \ldots + \alpha_s a_s) = \alpha_1 \phi(a_1) + \ldots + \alpha_s \phi(a_s).\]
\end{lem}

\begin{proof}

Consider a~linear combination \[v = \alpha_1 a_1 + \ldots + \alpha_s a_s\] that satisfies the~premise of~the~lemma. Without loss of~generality, let us assume that~$\alpha_1$ is the~largest among the~coefficients~$\alpha_i$. Put~$\alpha'_1 = \max(\alpha_1 - r_0, 0)$, and \[t = \alpha'_1 + \alpha_2 + \ldots + \alpha_s.\] Note that~$t \leq r - r_0$. Indeed, if~$\alpha_1$ is at~least~$r_0$, then the~claim is obvious. Otherwise we have $\alpha'_1 = 0$, and~$\alpha_i \leq \alpha_1 < r_0$ for~all~$i$, thus \[\alpha'_1 + \alpha_2 + \ldots + \alpha_s < (s - 1)r_0 = r - r_0.\]

Put~$u = \alpha'_1 a_1 + \alpha_2 a_2 + \ldots + \alpha_s a_s$. Let us construct a~path $u_0$,~\ldots,~$u_t$ in~$B_r$ from~$u_0 = 0$ to~$u_t = u$ as~follows. Start from~the zero sum~$u_0 = 0a_1 + \ldots + 0a_s$. To transfer from~$u_i$ to~$u_{i + 1}$, choose a~primary element~$a_j$ and put~$u_{i + 1} = u_i + a_j$, thus increasing~$j$-th coefficient of~the~sum by~1. At~each step we choose~$a_j$ in~such a~way that no~coefficient of~$u_{i + 1}$ exceeds the~corresponding coefficient of~$u$. The~process stops once~$u_i = u$ (obviously, this will take~$t$ steps).

Suppose that~after~$i$ steps we have~$u_i = \beta_1 a_1 + \ldots + \beta_s a_s$. Let us show by~induction that for~each~$i$ from~0 to~$t$ we have \[\phi(u_i) = \beta_1 \phi(a_1) + \ldots + \beta_s \phi(a_s).\] The~base case~$i = 0$ is trivial. Let us ensure correctness of~the~inductive step from~$i$ to~$i + 1$. Without loss of~generality, we will assume that~$u_{i + 1} - u_i = a_1$, then we have to~prove
\[\phi(u_{i + 1}) = (\beta_1 + 1) \phi(a_1) + \ldots + \beta_s \phi(a_s).\]
Note that~$\rho(u_i) \leq i \leq t \leq r - r_0$, thus by~Proposition~\ref{containment}~$B_{r_0} + u_i$ is a~subgraph of~$B_r$. Further, since~$u_0, \ldots, u_i$ is a~connected set of~vertices of~$\Gamma$, by Corollary~\ref{connballs} the restrictions~$\phi|_{B_{r_0} + u_i}$ and~$\phi|_{B_{r_0}}$ have the~same orientation, hence
\[\phi(u_{i + 1}) - \phi(u_i) = \phi(u_i + a_1) - \phi(u_i) = \phi(0 + a_1) - \phi(0) = \phi(a_1),\]
because~$a_1 \in \A'$. After~adding this to~the~representation of~$\phi(u_i)$ as~a~linear combination of~$\phi(a_1)$,~\ldots,~$\phi(a_s)$, we obtain the~induction claim for~$i + 1$, which proves the~claim for~all~$i$ from~0 to~$t$.

Now we have~$v - u_t = (\alpha_1 - \alpha'_1)a_1$, and~$\alpha_1 - \alpha'_1 \leq r_0$. We also have~$\rho(u_t) \leq t \leq r - r_0$, thus~$B_{r_0} + u_t$ is a subgraph of~$B_r$. Since the~restriction~$\phi|_{B_{r_0} + u_t}$ has the~same orientation as~$\phi|_{B_{r_0}}$, by~Proposition~\ref{ballmain} we obtain \[\phi(v) - \phi(u_t) = (\alpha_1 - \alpha'_1)\phi(a_1).\] After adding this to~$\phi(u_t) = \alpha'_1 \phi(a_1) + \ldots + \alpha_s \phi(a_s)$, we obtain the~claim of~the~lemma.

\end{proof}

Suppose that~$r \geq R$. Let~$A_r$ denote the~set of~vertices of~$B_r$ that are representable as~a~sum of~primary elements in~the~sense of~Lemma~\ref{mainsum}. Since for~any~$\phi \in \Aut(B_r)$ and~each primary~$a \in \A'$ the element~$\phi(a)$ is also primary, Lemma~\ref{mainsum} implies that~$\phi(A_r) = A_r$.

\begin{lem} \label{mainlinear}

There is an integer~$R_L$ such that for~each integer~$r \geq R_L$ and each~$\phi \in \Aut(B_r)$ there is a~linear map~$T: \R^k \to \R^k$ with~$\phi(x) = T(x)$ for~all~$x \in A_r$.

\end{lem}

\begin{proof}

Since~$\A'$ has full rank in~$\Q^k$, we can choose a~basis of~$\Q^k$ among elements of~$\A'$; let~$a_1$,~\ldots,~$a_k$ denote elements of~the~basis. Any element~$a \in \A'$ is representable as a~rational linear combination of~$a_1, \ldots, a_k$: \[a = \alpha_1 a_i + \ldots + \alpha_k a_k.\] After multiplying by a common denominator, we obtain: \[Ka = \beta_1 a_1 + \ldots + \beta_k a_k,\] where~$K$ and all~$\beta_i$ are integers. Put~$K'_a = \max(R, K + |\beta_1| + \ldots + |\beta_k|)$. Invoking Lemma~\ref{mainsum} for~a~ball of~radius~$r \geq K'_a$ and the~equality
\[Ka - \beta_1 a_1 - \ldots - \beta_k a_k = 0,\]
we obtain that
\[K \phi(a) - \beta_1 \phi(a_1) - \ldots - \beta_k \phi(a_k) = \phi(0) = 0.\]
Division by~$K$ and a~transfer to~the~right-hand side yields
\[\phi(a) = \alpha_1 \phi(a_i) + \ldots + \alpha_k \phi(a_k).\]

Choose~$R_L$ as the~largest value of~$K'_a$ for~all~$a \in \A'$, then the~linear map~$T$ induced by the values of~$\phi \in \Aut(B_r)$ on~vectors~$a_1$, \ldots, $a_k$ agrees with~$\phi$ on~all elements of~$\A'$ once~$r \geq R_L$, thus by~Lemma~\ref{mainsum} it must agree with~$\phi$ on~all elements of~$A_r$.

\end{proof}

Let~$\overline{\A'} = \A \backslash \A'$ denote the~set of~all secondary elements of~$\A$.

\begin{prp} \label{ballnotmain}

If~$\phi \in \Aut(B_r)$, and~$v \in V(B_r)$ satisfies~$\rho(v) \leq r - r_0$, then for~each~secondary element~$a \in \A$ the element $\phi(v + a) - \phi(v)$ is secondary as well.

\end{prp}

\begin{proof}

The~graph~$B_{r_0} + v$ is a~subgraph of~$B_r$, hence by~Proposition~\ref{ballmain} the~map~$a \to \phi(v + a) - \phi(v)$ permutes~$\A'$, thus it must also permute~$\overline{\A'}$.

\end{proof}

Let~$B'_r = B^{\overline{\A'}}_{r}$ denote the~ball of~radius~$r$ in~the~Cayley graph~$\Gamma'$ of~the~additive group generated by~$\overline{\A'}$. Trivially, for~each~$r$ the~graph~$B'_r$ is a~subgraph of~$B_r = B^{\A}_r$. We will write~$\rho'(x) = \rho^{\overline{\A'}}(x)$ for~the~length of~the~shortest~path between~0 and~$x$ in~the~graph~$\Gamma'$.

\begin{lem} \label{autnotmain}

Suppose that~$0 \leq r' = r - r_0$. Then the~restriction of~any auto\-mor\-phism~$\phi \in \Aut(B_r)$ to~$B'_{r'}$ is an~automorphism of~$B'_{r'}$.

\end{lem}

\begin{proof}

We have to~prove~$\phi(v) \in V(B'_{r'})$ for~all~$v \in V(B'_{r'})$, and~$\phi(v)\phi(u) \in E(B'_{r'})$ for all~$vu \in E(B'_{r'})$.

Suppose that~$v \in V(B'_{r'})$. Consider a~shortest path in~$B'_{r'}$ from~0 to~$v$, denote its vertices~$u_0$, \ldots, $u_t$ with~$u_0 = 0$,~$u_t = v$,~$t \leq r'$, and~$u_{i + 1} - u_i \in \overline{\A'}$ for~all integer~$i$ from~0 and~$t - 1$. For each vertex~$u_i$ we have \[\rho(u_i) \leq \rho'(u_i) \leq t \leq r - r_0,\] hence Proposition~\ref{ballnotmain} implies~$\phi(u_{i + 1}) - \phi(u_i) \in \overline{\A'}$. Thus there exists a~path~$\phi(u_0)$,~\ldots, $\phi(u_t)$ between~$\phi(u_0) = 0$ and~$\phi(u_t) = \phi(v)$ that has length at~most~$r'$ and consists exclusively of~transitions by~elements of~$\overline{\A'}$, therefore~$\phi(v) \in V(B'_{r'})$.

Similarly, consider an~edge~$xy \in E(B'_{r'})$. The~above reasoning implies that~$\phi(x), \phi(y) \in V(B'_{r'})$. Since~$xy$ is an edge of~$B'_{r'}$, we have that~$y - x \in \overline{\A'}$. By invoking Proposition~\ref{ballnotmain} once again we obtain that~$\phi(y) - \phi(x) \in \overline{\A'}$, hence~$\phi(x) \phi(y) \in E(B'_{r'})$.

\end{proof}

The~last component of~the~proof of~Theorem~\ref{finiteball} is the~following

\begin{prp} \label{deltanotmain}

There exists an integer~$\Delta$ such that for each vertex~$v \in G$ there is a~shortest path in~$\Gamma$ from~0 to~$v$ that contains at~most~$\Delta$ transitions by~elements of~$\overline{A'}$.

\end{prp}

\begin{proof}

Let us recall that by~Proposition~\ref{mainshortcut} for~any element~$a \in \A$ there exists an~integer~$K_a$ and a~repre\-sen\-ta\-tion~\[K_a a = \alpha_1 a_1 + \ldots + \alpha_s a_s,\] where~$a_i$ are primary elements, and~$\alpha_i$ are non-negative integers which sum does not exceed~$K_a$. Now, consider any shortest path from~0 to~$v$. Let~$a$ be an arbitrary secondary element. If the~path contains at~least~$K_a$ transitions by~$a$, we can interchange them by~$\alpha_1$ transitions by~$a_1$, \ldots, and~$\alpha_s$ transitions by~$a_s$. The~length of~the~path will not change since the~sum of~$\alpha_i$ is at~most~$K_a$ and the~path was one of~the~shortest.

Perform all possible replacements of~this kind for~all secondary elements. The~resulting path is one of~the~shortest in~$\Gamma$ from~0 to~$v$ and contains at~most~$\Delta = \sum_{a \in \overline{\A'}} K_a$ transitions by~secondary elements.

\end{proof}

\subsection{Proof of~Theorem~\ref{finiteball}}

We are now ready to~prove Theorem~\ref{finiteball}.

\begin{proof}
First, we will prove that all automorphisms of sufficiently large balls are linear. We will use induction by the size of $\A$. The~convex hull of~a~non-empty set~$\A$ can't have an~empty set of~vertices, hence~$|\overline{\A'}| < |\A|.$

If~$\A' = \A$, then the~claim follows immediately from~Lemma~\ref{mainlinear} with~$R(\A) = R_L$. Otherwise, put
\[R_1 = \max(r_0, R_L)\text{ (see Lemmas~\ref{mainsum} and~\ref{mainlinear})},\]
\[R_2 = \max\left(R(\overline{\A'}), \Delta,\max_{x \in \overline{\A'}} K_x\right)\text{ (see Propositions~\ref{deltanotmain} and~\ref{mainshortcut})},\]
and~$R(\A) = R_1 + R_2$. Suppose that~$r \geq R(\A)$ and~$\phi \in \Aut(B_r)$. By Lemma~\ref{mainlinear} there is a~linear map~$T: \R^k \to \R^k$ that agrees with~$\phi$ on~the~set~$A_r$. Denote~$B' = B^{\overline{\A'}}_{R_2}$. Since~$r_0 + R_2 \leq R_1 + R_2 \leq r$, Lemma~\ref{autnotmain} implies that the~restriction of~$\phi$ to~$V(B')$ is an automorphism of~$B'$. Moreover, let~$U$ denote the~linear subspace of~$\R^k$ spanned by~elements of~$V(B')$. Since~$R_2 \geq R(\overline{\A'})$, then by the induction hypothesis there is a~linear map~$T': U \to U$ that agrees with~$\phi$ on~$V(B')$.

Let us prove that $T$ agrees with $T'$ on $V(B')$. For~an~arbitrary~$x \in \overline{\A'}$ we have~$T'(x) = \phi(x)$. By~Proposition~\ref{mainshortcut} we have that~$K_x x$ can be represented as~a~sum of~at~most~$K_x$ primary elements. We also have~$R \geq K_x$, hence~$K_x x \in A_r$. But this implies \[T(x) = \phi(K_x x) / K_x = T'(x).\] Finally, $T'$ is uniquely determined by~its values on~elements of~$\overline{\A'}$, con\-se\-quent\-ly,~${T|_U = T'}$. It suffices to~notice that~$V(B') \subset U$.

Lastly, we will prove that~$\phi$ agrees with~$T$ on~all other vertices of~$B_r$. Choose an~arbitrary~$x \in V(B_r)$. By Proposition~\ref{deltanotmain} there exists a~shortest path in~$\Gamma$ from~0 to~$x$ that contains at most~$\Delta$ transitions by~elements of~$\overline{\A'}$; clearly, this path lies completely inside~$B_r$. Let us represent~$x = a + b$, where~$a$ is the~sum of~all transitions by~elements of~$\A'$, and~$b$ is the~sum of~all transitions by~elements of~$\overline{\A'}$. Further, let~$s$ denote the~number of~transitions by~elements of~$\overline{\A'}$. Since~$s \leq \Delta \leq R_2$, then~$b \in V(B')$, hence~$\phi(b) = T'(b) = T(b)$. The~graph~$B_{R_1} + b$ is a~subgraph of~$B_r$ since
\[R_1 + \rho(b) \leq R_1 + \Delta \leq R_1 + R_2 \leq r.\]
By~Lemma~\ref{mainlinear} we have that the~automorphism~$\phi_b \in \Aut(B_{R_1})$ defined by~the~formula~\[\phi_b(x) = \phi(b + x) - \phi(b)\] is linear on~elements of~$A_{R_1}$. Further, it has the same orientation as~$\phi$ since vertices~0 and~$b$ are connected. Consequently, we have~$\phi_b(x) = T(x)$, and~\[\phi(b + x) = T(b) + T(x) = T(b + x).\] Thus, linearity is established.

To~finish the~proof, we have to show that all automorphisms of large enough balls are extendable to automorphisms of $\Gamma$. We have that~$r \geq R(\A)$ implies linearity of~all automorphisms of~$B_r$. Every linear automorphism is uniquely determined by~its values on~elements of~$\A$; moreover, these values must form a~permutation of~$\A$. Suppose that~$\phi \in \Aut(B_r)$, and for~the~linear map~$T$ induced by~$\phi$ we have~$T(\Gamma) \neq \Gamma$. Then~the~graph~$T(\Gamma)$ must differ from~$\Gamma$ in a~vertex or an~edge at~distance~$d > r$ from~the~vertex~0. To eliminate~$\phi$ from~$\Aut(B_r)$, increase the~value of~$R(\A)$ to~$d$. Since there are only finitely many permutations of~$\A$, this process will require a~finite number of~steps, after which the~final value of~$R(\A)$ that satisfies Theorem~\ref{finiteball} is produced.
\end{proof}

Combining Theorem~\ref{finiteball} with Corollary~\ref{connballs}, we obtain a~``ball bundle'' version of~Theorem~\ref{finiteball}.

\begin{corollary} \label{finitebundle}

In assumptions of Corollary~\ref{connballs}, if we additionally have $r \geq R(\A)$, then there is a unique~non-degenerate affine map~$T: \R^k \to \R^k$ such that~$\phi(x) = T(x)$ for all~$x \in V(C)$. Furthermore,~$T$ defines~an~automorphism of~$\Gamma$.

\end{corollary}

\section{Strict and/or injective $\A$-embeddability, $k \geq 2$} \label{zk}

\subsection{The {\bf LOGIC-ENGINE} problem} \label{enginedesc}

The~term \emph{``logic engine''} (coined by~Eades and Whitesides in~\cite{eades}) refers to~a~certain type of~a~geometric setup that is designed to~``mechanically'' emulate solution of~the~{\bf NAE-3-SAT} {\it (not-all-equal 3-satisfiability)} problem. The~earliest construction of~this kind was used by~Bhatt and Cosmadakis (see~\cite{bhatt}) to~establish NP-hardness of~embedding a~graph in~the~square grid with~unit-length edges. Since then, similar setups were employed in~a~number of~papers concerned with~complexity of~geometric problems (\cite{gregori,eades,eadestrees,fekete,kitching}).

Let us give a~rough description of~a~planar logic engine setup. An~{\it axle} is rigidly mounted on~a~rigid {\it frame} (a~{\it rigid} configuration is~one that allows unique realization up to~isometry). There are~$n$ straight rigid {\it rods} attached to~the~axle at~$n$ equidistant points; the~rods extend to~the~both sides of~the~axle and always stay perpendicular to~it. Naturally, each rod has two possible directions relative to~the~axle, and directions of~different rods are independent. Additionally, each rod has several straight rigid {\it flags} attached to~it that must stay perpendicular to~the~rod and thus have two possible directions each. Flags can be attached to~the~rods on~$m$ different levels at~both sides of~the~axle.

\begin{center}
\begin{figure}[htb]
\centering
\includegraphics{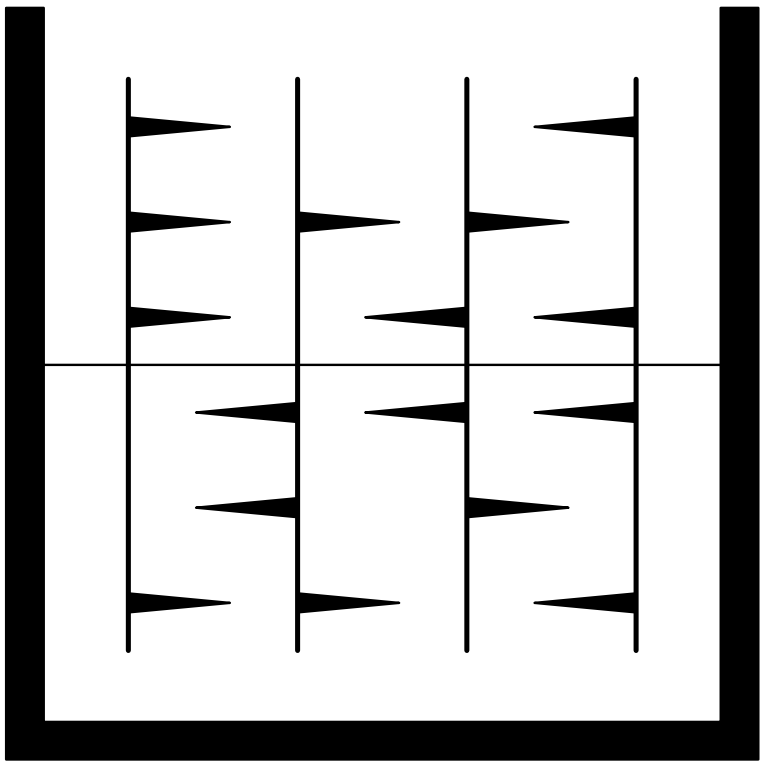}
\caption{An~illustration of~a~logic engine with~$n = 4$ rods and~$m=3$ flag levels on~both sides of~the~axle}
\label{enginepic}
\end{figure}
\end{center}

The~length of~flags is adjusted so that adjacent flags on~the~same level will collide if pointed towards each other. Also, the~frame prevents the~flags attached to~the~outer\-most rods from~pointing outwards.

The~structure of~a~logic engine is defined by~the~numbers~$n$ and~$m$, along with~$2nm$ numbers~$a_{ijk}$ chosen from~$\{0, 1\}$ that describe whether a~flag is attached to~$i$-th rod on~$j$-th level on~$k$-th side of~the~axle. A~logic engine is {\it realizable} if directions of all rods and flags can be chosen so that no~flag collides with~another flag or the~frame. We pose the~decision problem {\bf LOGIC-ENGINE} of~determining realizability of~a~given logic engine.

\begin{prp}

{\bf LOGIC-ENGINE} is NP-complete.

\end{prp}

\begin{proof}

Clearly, the~problem is in~NP since logic engine realizability is easily certifiable with a small certificate. NP-hardness of~the~problem is shown in, say,~\cite{bhatt} and~\cite{eades} by~reducing NP-hard problem~{\bf NAE-3-SAT} to~{\bf LOGIC-ENGINE}.

\end{proof} 

Our goal is to~reduce~{\bf LOGIC-ENGINE} to~strict and/or injective~$\A$-embeddability in~$\R^1$. We will do this by~explicitly constructing a~graph~$X$ such that a~given logic engine is realizable if and only if~$X$ is isomorphic to~a~subgraph of~$\Gamma$.

\subsection{Choice of basis in $\Z^k$}

Before we proceed, let us find a convenient coordinate system in $G \sim \Z^k$. If~$a_1, \ldots, a_m \in \R^n$, then let~$(a_1, \ldots, a_m)$ denote the~$n \times m$ matrix which columns contain coordinates of~the~vectors~$a_1$, \ldots, $a_m$. We will call a~collection of~vectors~$a_1, \ldots, a_m \in \R^n$ \emph{non-degenerate} if the~matrix~$(a_1, \ldots, a_m)$ has rank equal to~the~di\-men\-sion~$n$.

\begin{lem} \label{basis}

Suppose that~$A = \{a_1, \ldots, a_m\} \subset \R^n$ is a~non-degenerate collection of~vectors. Then it is possible to~choose a~basis~$b_1, \ldots, b_n$ of~$\R^n$ among~elements of~$A$ so that for~each~$a \in A$ all coefficients~$\alpha_i$ of~the~unique representation~$a = \sum \alpha_i b_i$ do not exceed~1 by~absolute value.

\end{lem}

\begin{proof}

Choose~$b_1, \ldots, b_m$ so that~$|\m{det}(b_1, \ldots, b_m)|$ is largest possible, and put~$B = (b_1, \ldots, b_m)$. Since~$A$ is a~non-degenerate collection, we must have~$\m{det}B \neq 0$. Coordinates of~a~vector~$a\in A$ with~respect to~the~basis~$b_1, \ldots, b_m$ are defined by~the~unique vector~$x \in \R^n$ that satisfies the~equation~$Bx = a$. By~Cramer's rule we~have \[x_i = \frac{\m{det}(b_1, \ldots, b_{i - 1}, a, b_{i + 1}, \ldots)}{\m{det}B}.\] Maximality of~$|\m{det}B|$ implies~$|x_i| \leq 1$, thus the lemma is proven.

\end{proof}

Let $\B$ denote the~basis chosen from~the~set~$\A'$ via~Lemma~\ref{basis}. In~the~sequel, coordinates of~all elements of~$G$ will be considered exclusively with~respect to~the~basis~$\B$. We also define the~{\it norm~$||x||$} of~an~element~$x \in G$ as~the~value of~$l_{\infty}$-norm of~$x$ with~respect to~$\B$ (here we override the~definition of~norm introduced in~Chapter~\ref{balls}).

\subsection{Balls locality and solidity} \label{locality}

Let~$S_a = [-a, a]^k \subset \R^k$ denote the~hypercube with~side length~$2a$ centered at~the~origin.

By construction, the basis~$\B$ consists of primary elements. Lemma~\ref{basis} implies that all elements of $\A'$ belong to $S_1 = [-1, 1]^k$, hence by convexity all elements of $\A$ belong to $S_1$. It follows that vertices of $B_r$ are confined to $S_r$.

The~constructions will consist of~$\Gamma$-rigid ball bundles that correspond to~independent rigid components connected via~auxiliary edge chains. Following the~mechanical analogy, we expect the~bundles to~behave like physical objects, for~instance, different bundles should not be able to~collide in~any injective embedding. This is not generally the~case since vertices of~different bundles may permeate each other. However, large enough balls have ``solid zones'' that never collide for~disjoint balls:

\begin{prp} \label{black}

Suppose that~$r \geq \max_{||x|| \leq 1}\rho(x)$. Then in~any injective embedding~$\phi$ of~$(B_r + z) \sqcup (B_r + z')$ in~$\Gamma$ with $B_r + z$, $B_r + z'$ disjoint, we must have $(S_1 + \phi(z)) \cap (S_1 + \phi(z')) = \varnothing$.

\end{prp}

\begin{proof}

Suppose that there is a~point~$x \in (S_1 + \phi(z)) \cap (S_1 + \phi(z'))$. We can choose~$x$ to be a~vertex of~the~hypercube~$S_1 + \phi(z)$. In~this case,~$x \in V(\phi(B_r + z))$, in~particular,~$x \in G$. Since we also have~$||x - \phi(z')|| \leq 1$, we must have~$x \in V(\phi(B_r + z'))$, but the~balls~$\phi(B_r + z)$ and~$\phi(B_r + z')$ must be disjoint since~$\phi$ is injective, contradiction.

\end{proof}

Note that $\max_{||x|| \leq 1}\rho(x)$ is taken over a finite set of points of $\Z^k$, and is, therefore, well-defined. In the sequel, we put \[r = \max\left(R(\A), \max_{||x|| \leq 1} \rho(x)\right).\]

We will use the~following convention when discussing ball bundles that build up parts of~a~construction. To~each ball~$B_r + z$ we assign a~``\emph{black}'' region~$S_1 + z$ and a~``\emph{gray}'' region~$S_r + z$. ``Black'' regions of~(balls of)~different bundles can not intersect in~any injective embedding. On~the~other hand, bundles can only obstruct each others' injective embedding when their ``gray'' regions intersect. Graphically, ``black'' regions will be colored with~dark gray, and ``gray'' regions will be colored with~light gray.

\begin{center}
\begin{figure}[h]
\centering
\begin{subfigure}[b]{0.9\textwidth}
\centering
\includegraphics[scale=0.7]{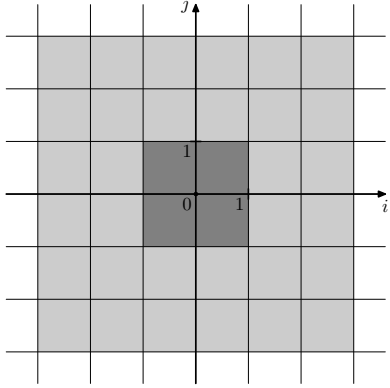}
\end{subfigure}
\captionsetup{justification=centering,margin=2cm}
\caption{``Black'' and ``gray'' regions of a ball of radius 3 centered at the origin}
\end{figure}
\end{center}

\subsection{Construction for~the~case~$k = 2$} \label{z2}

Suppose that~$k = 2$. Let~$i$ and~$j$ denote the~two elements of~$\B$.

\subsubsection{Construction outline}

\begin{center}
\begin{figure}[htbp]
\centering
\begin{subfigure}{\textwidth}
\centering
\includegraphics[width=84mm]{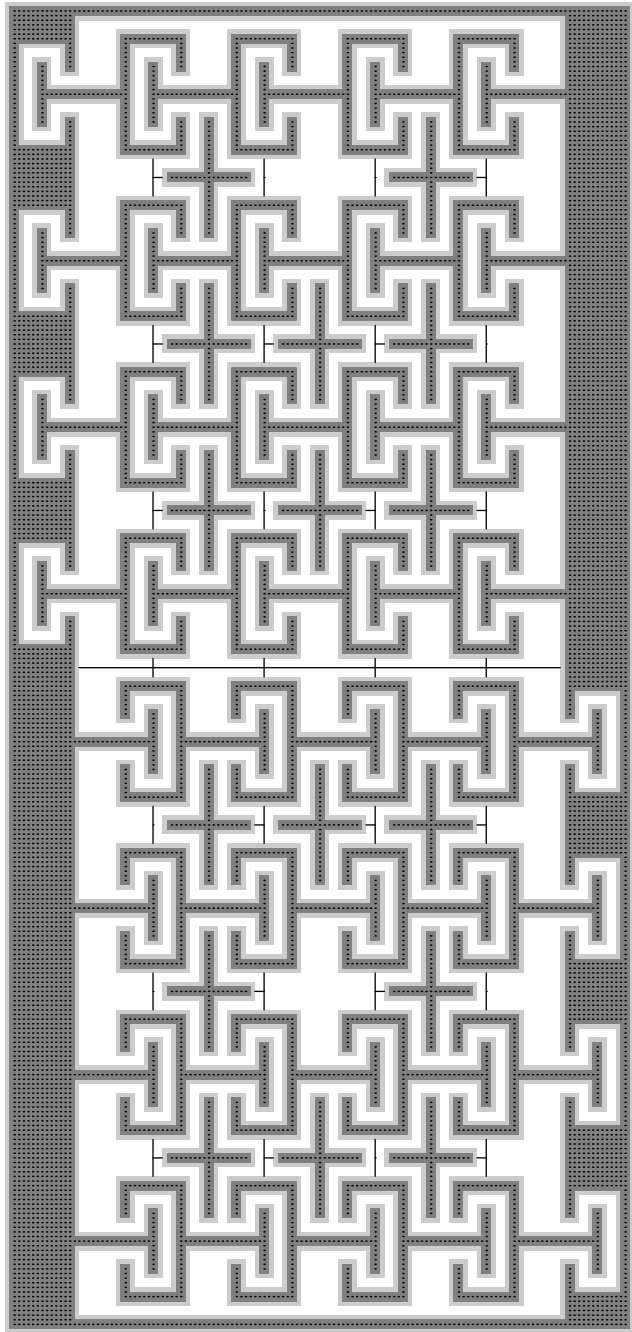}
\end{subfigure}
\captionsetup{justification=centering,margin=2cm}
\caption{Schematic picture of~the~graph~$X$ constructed by~the~logic engine pictured on~Fig.~\ref{enginepic}. Bold dots denote balls' centers, line segments denote auxiliary chain edges.}
\label{schemepic}
\end{figure}
\end{center}

Let numbers $a_{ijk}$ describe a logic engine with~$n$ rods and~$m$ flag levels. We will construct a graph $X$ such that $X$ is (strictly/non-strictly) injectively embeddable in~$\R^1$ if and only if the logic engine is realizable. As we said before, the construction will consist of bundles of balls of radius $r$ and auxiliary edge chains between them.

The \emph{skeleton}~$\Theta$ will consist of a~\emph{frame}~$O$ and an~\emph{axle}~$C$. The~frame~$O$ is a~bundle of~balls centered at~integer points of~the~rectangular border $[0; W] \times [-H; H]$ (with parameters~$W$ and~$H$ to be chosen later). We will later modify~$O$ by adding \emph{docking components} (see below) as shown on Figure~\ref{schemepic}. The~axle~$C$ is~a~path~$(0, 0)$, $(1, 0)$, \ldots, $(W, 0)$. Finally,~$\Theta$ is induced by~$V(O) \cup V(C)$.

Next, each of~the~$n$ rods of~the~logic engine will be represented by~two \emph{chains} anchored at~the~same point on~the~axle. Each chain consists of~$m + 1$~\emph{links}. Each link is a~bundle of~balls centered at~integer points of~the~rectangle~$[0, w] \times [0, h]$ with addition of~docking components (parameters~$w$ and~$h$ will be chosen later as~well). We add edge chains of certain length to~connect adjacent links, as~well as~the~first link in~the~chain with~the~anchor point on~the~axle.

Finally, we attach \emph{flags} to~each chain according to~flags positions in the logic engine. Each of the flags is a~ball bundle with~centers in~integer points of~the~``cross''
\[\{(x, 0) \mid x \in [-w_f; w_f]\} \cup \{(0, y) \mid y \in [-h_f; h_f]\}.\]
Flags can be attached at~$m$ possible locations at~midpoints of~edge chains between~adjacent links. Each flag points to~the~left or to~the~right in any embedding. Flags' positions satisfy two rules:

\begin{itemize}

\item no flag can point ``into~the~wall'',

\item if two flags are attached on~the~same level on~adjacent chains, then in~any injective embedding such that their chains are on~the~same side on~the~axle, the~flags must not point towards each other.

\end{itemize}

All the~restrictions described above guarantee that injective embeddability of~the~con\-structed graph~$X$ is equivalent to~realizability of~the~logic engine.

In~what follows we describe all parts of~the~construction in~detail, and also provide a~way to~choose parameters $W$, $H$, $w$, $h$, $w_f$, $h_f$ so that injective embeddings of~$X$ behave as~expected. We will write~$O(1)$ for~any value with~absolute value bounded by~a~constant independent of~$n$ and~$m$ (in particular, note that~$r = O(1)$). Also,~we will write~$c_i$ for~certain constants when explicit value is unimportant.

\subsubsection{Link-axle and link-link connections} \label{links}

Let $\phi$ be any injective embedding of~the~graph~$X$ in~$\Gamma$, and~$T$ be the~affine automorphism of~$\Gamma$ induced by~the~restriction~$\phi|_{O}$ by~Corollary~\ref{finitebundle}. Then~$T^{-1} \phi$ must also be an~injective embedding of~$X$, hence we can assume that~$T \equiv 1$, and~$\phi$ acts identically on~$O$. Furthermore, $\phi$ must also act identically on $C$; indeed, all edges of $C$ correspond to the primary direction $i$, hence there is a~unique shortest path between~$(0, 0)$ and~$(W, 0)$, and~$\phi$ must preserve its vertices.

\begin{center}
\begin{figure}[htb]
\captionsetup{justification=centering}  
\centering
\begin{subfigure}[t]{0.5 \textwidth}
\centering
\includegraphics[width=39mm]{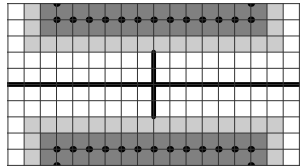}
\caption{$j \pm i \not \in \A$ case}
\end{subfigure}
\begin{subfigure}[t]{0.5 \textwidth}
\centering
\includegraphics[width=39mm]{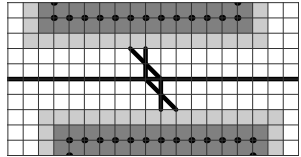}
\caption{$j - i \in \A$ case}
\end{subfigure}
\caption{Link-axle connection}
\label{linkaxispic}
\end{figure}
\end{center}

Consider the~connection between a~link of~width~$w$ and~the~axle via an auxiliary two-edge chain. To~ensure possibility of~a~strict embedding, the~structure of~the~joint will depend on~whether the~elements~$j \pm i$ belong to~$\A$ (note that~$j + i$ and~$j - i$ cannot belong to~$\A$ simultaneously since that would imply that~$j$ is not a~vertex of~$\Conv \A$ and, therefore, not primary). Fig.~\ref{linkaxispic} depicts two possible link-axle joints along with~their embeddings (the~$j + i \in \A$ case is symmetrical to~the~(b) case). Note that in both cases the embeddings are locally induced, that is, there are no hidden edges between the axle and the edge chain. We also have to~add all possible edges between the~link and the~chain vertices. Note that these edges can only be incident to~the~midpoint of~the~chain since we have~$||x|| \leq 1$ for~all~$x \in \A$.

\begin{center}
\begin{figure}[h]
\begin{subfigure}[t]{0.5\textwidth}
\centering
\includegraphics[width=39mm]{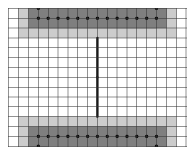}
\captionsetup{justification=centering}  
\caption{Link-link connection (when~$j \pm i \not \in \A$)}
\end{subfigure}
\begin{subfigure}[t]{0.5\textwidth}
\centering
\includegraphics[width=39mm]{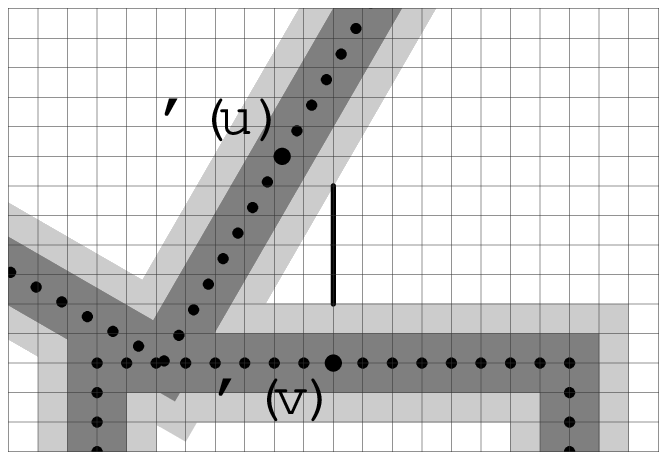}
\captionsetup{justification=centering}  
\caption{Obstructed embedding of~adjacent links in~the~$i_1 \neq \pm i_2$ case after increasing~$w$}
\label{linkaxiscrosspic}
\end{subfigure}
\caption{Link-link connection}
\end{figure}
\end{center}

Adjacent links will be connected via~edge chains of~length~$l = 2r + 4$ in a~similar fashion. Let us show that it is possible to~choose~$w$ in~such a~way that in~any injective embedding the~horizontal edges of~all chains are parallel to~each other.

Suppose that~$L_1$ and~$L_2$ are adjacent links in a~chain, and the~link chain is located to~the~top of~the~axle (the~situation is completely symmetrical at~the~bottom). Let~$\phi_1$ and~$\phi_2$ be the~affine automorphisms of~$\Gamma$ defined by~restrictions of~$\phi$ to~$L_1$ and~$L_2$ respectively. Suppose that~$\phi_1$ and~$\phi_2$ map the~vector~(not the~point!)~$i$ to~vectors~$i_1$ and~$i_2$ respectively, furthermore, suppose that $i_1 \neq \pm i_2$. Since~$i_1$ and~$i_2$ are primary directions, they must be vertices of~$\Conv \A$, thus~$i_1 \neq \pm i_2$ must imply that~$i_1$ and~$i_2$ are not parallel.

Consider the~ball centers~$u_t = u + t i$ of~the~``top edge'' of~$L_1$, and ball centers~$v_t = v + t i$ of~the~``bottom edge'' of~$L_2$, where~$t$ is an~integer parameter that ranges within~$[-w; w]$ (see Fig.~\ref{linkaxiscrosspic}). Their images under the~embedding~$\phi$ lie on~two non-parallel lines~$\phi(u) + ti_1$ and~$\phi(v) + ti_2$. Let us find the~intersection point of~these lines:~$\phi(u) + t_1 i_1 = \phi(v) + t_2 i_2$, with $t_1$, $t_2$ not necessarily integer. Next, we can find integer~$t'_1$ and~$t'_2$ near~$t_1$ and~$t_2$ respectively, such that $||\phi(u_{t'_1}) - \phi(v_{t'_2})|| \leq 1$. If we introduce the requirement~$w > \max(|t'_1|, |t'_2|)$, then the~black regions of~balls centered at~$\phi(u_{t'_1})$ and~$\phi(v_{t'_2})$ must intersect, hence this injective embedding becomes forbidden.

Finally, note that we must have~$||\phi(v) - \phi(u)|| \leq l$, hence~$\max(|t'_1|, |t'_2|) < l \alpha$, where~$\alpha = \alpha(i_1, i_2)$ depends only on~$i_1$ and~$i_2$. That is, to~forbid injective embeddings with~$i_1 \neq \pm i_2$ it suffices to require~$w > l\alpha_{max}$, where~$\alpha_{max}$ is the~maximal value of~$\alpha(i_1, i_2)$ among all pairs~$i_1, i_2 \in \A'$ with~$i_1 \neq \pm i_2$. Since~$l = O(1)$, the~obtained requirement can be written as \begin{align} \label{linkshor} w > c_0. \end{align}

A~similar reasoning applied to~the~link-axle connection implies that in any injective embedding with~$i_1 \neq \pm i$ the~bottom edge of~the~first link will be obstructed by~the~axle vertices, therefore we must have~$i_1 = \pm i$. Hence, under requirement~\eqref{linkshor} all horizontal edges of~all links must preserve their orientation. Let us note that we haven't yet ensured that horizontal edges can't have opposite directions, that is, the~$i_1 = -i_2$ case is not ruled out yet.

Let us choose the~distance $\delta$ between~adjacent link-axle connection points to~be equal to~$w + 4r + 4 = w + O(1)$. Further, let the~distance from the extreme link-axle connection points and the frame vertices be at least $\delta / 2$, then it is possible to place the~links without colliding with~the~frame. Put~$W = (n + 2) \delta$, then a~rectangle of~width~$W$ fits all~$n$ links as~well as~docking components (see Fig.~\ref{schemepic}). Similarly to~the~choice of~$w$, we can show that it is possible to~choose large enough \begin{align} \label{linksver} h > \beta W, \end{align} such that in~any embedding all vertical edges of~all links are vertical or collide with the~frame.

Furthermore, let us choose~$H$ so that to~fit all link chains inside~the~frame vertically. It can be verified that the~height of the~top link's gray region is at most~$2 + h + 2r$ (height of~the~gray region of~the~bottom link) + $m \times (h + l + 2r)$ (vertical translations between adjacent links). Putting~$H = 2 + h + 2r + m(h + l + 2r) + r + 2$, we guarantee that the~distance between gray regions of~the~top side of~the~frame and the~top link is at~least~2. Clearly,~$H = O(mh)$.

\subsubsection{Docking components}

\begin{center}
\begin{figure}[htb]
\centering
\begin{subfigure}{\textwidth}
\centering
\includegraphics[width=129mm]{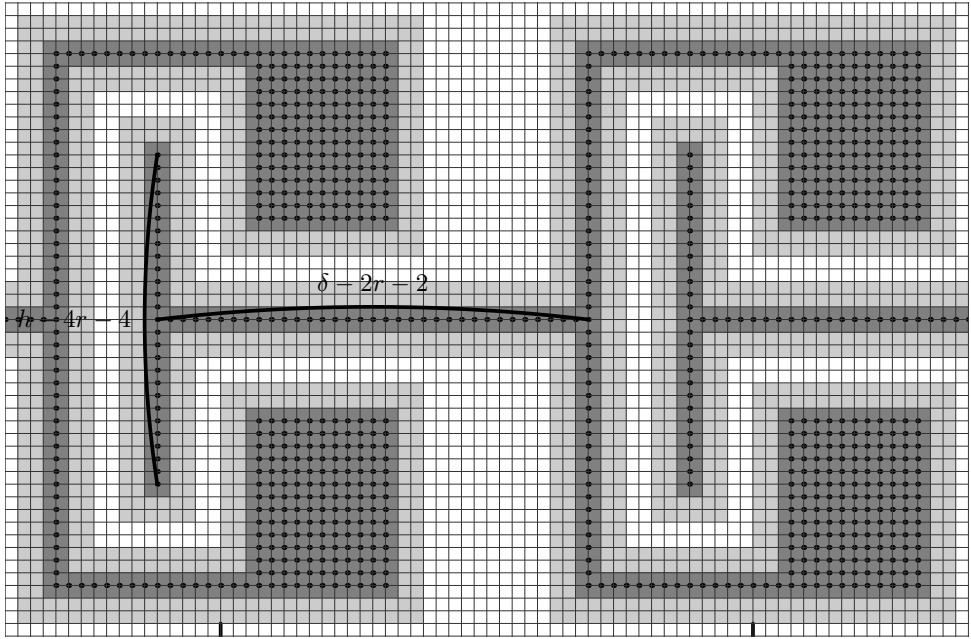}
\end{subfigure}
\captionsetup{justification=centering}
\caption{Docking component}
\label{couplingpic}
\end{figure}
\end{center}

For~the~construction to~behave properly under~possible embeddings, we must ensure that adjacent flags on~the~same level and side (relative to~the~axle) are located at~close height in~any embedding. However, we have ``flexible'' connections between adjacent links in~a~chain, thus there can be a~significant discrepancy in~height of~neighbouring links and flags. Moreover, these discrepancies can add up without a~limit since the~number of~levels~$m$ is unbounded. We will introduce \emph{docking components} in each of~the~links and in~the~frame to~ensure that the~height discrepancy of~adjacent links stays bounded.

Consider a~pair of~adjacent links located on~the~same level to~the~top of~the~axle. Let us modify each of~the~links by~expanding the~bundles with~extra ball centers lying on~a~forked T-shaped~\emph{``antenna''} (see Fig.~\ref{couplingpic}). Further, let us make a~\emph{``receiver''} in~each of~the~links by~removing all balls with~gray region within distance~1 of~the~gray region of~the~adjacent link's antenna. We will choose dimensions of~antennas as~large as~possible so that a~link remains connected after making a~suitable receiver.

As shown above, in~any injective embedding all links must preserve orientation under require\-ments~\eqref{linkshor} and~\eqref{linksver}. Our intention is that in~any embedding each pair of~neighbouring links must be interlocked, that is, each~antenna must be inside the~cor\-res\-ponding receiver.

For a~pair of~neighbouring links~$L_1$ and~$L_2$ to~the~top of~the~axle let us consider vertices~$u_1$,~$v_1$,~$u_2$,~$v_2$ (see Fig.~\ref{couplingpic2}) that are the~endpoints of~the~chains connecting~$L_1$ and~$L_2$ with~the~axle or with a~previous link; in~the~latter case we will assume that the~previous links are interlocked. We can verify that in any case~$\phi(v_2) - \phi(v_1) = (\delta + O(1), O(1))$ holds.

\begin{center}
\begin{figure}[htb]
\centering
\begin{subfigure}{\textwidth}
\centering
\includegraphics{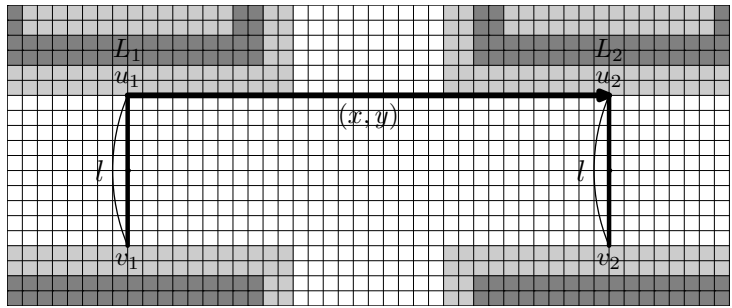}
\end{subfigure}
\captionsetup{justification=centering}
\caption{Edge chains connecting adjacent link pairs}
\label{couplingpic2}
\end{figure}
\end{center}

First, suppose that the~horizontal edges of~$L_1$ and~$L_2$ have the~same direction, in that case we must have~$\phi(L_2) = \phi(L_1) + (x, y)$. If~$L_1$ and~$L_2$ are not interlocked in~$\phi$, then we have either~$x > \delta + w + c_1$ or~$|y| > h + c_2$. But
\[(x, y) = \phi(u_2) - \phi(u_1) = (\phi(u_2) - \phi(v_2)) - (\phi(u_1) - \phi(v_1)) + (\delta + O(1), O(1)).\] Note that~$u_1$ and~$v_1$ are connected by~a~path of~length~$l$; the~same holds for~vertices~$u_2$ and~$v_2$. Consequently, $x \leq 2l + \delta + c_3$ and $|y| < 2l + c_4$. Let us enforce the~inequalities~$2l + \delta + c_3 < \delta + w + c_1$ and~$2l + c_4 < h + c_2$ by~increasing~$w$ and~$h$ (if necessary), so that none of~the~cases corresponding to~non-interlocked links~$L_1$ and~$L_2$ are possible. Since~$l = O(1)$, the~two restrictions can be written as
\begin{align} \label{coupling1} w > c_5, h > c_6\end{align}
Under these restrictions,~$L_1$ and~$L_2$ must be interlocked. An~inductive argument implies that all corresponding pairs of~links will be interlocked (the~situation is symmetrical to~the~bottom of~the~axle).

Finally, we must consider a~situation when horizontal edges of~links may have opposite direcion. In this case, we must have two links (or a link and a frame-attached docking component) with antennas pointing towards each other. Following the~notation of~the~previous paragraph, in this case we must have either~$x > 2w + c'$ or~$|y| > h / 2 + c''$ for~certain constants$~c', c''$. This situation can be eradicated by~strenghtening the~requirement~\eqref{coupling1} if necessary.

Let us note that the shift between any pair of~adjacent links in~a~chain (to~the~top of~the~axle) is~$(O(1), h + O(1))$ since they are connected by~a~path of~length~$l = O(1)$.

\subsubsection{Flags and attachments}

A~flag is attached to~the~midpoint of an auxiliary edge chain between consecutive links as shown on Fig.~\ref{flagspic1} (similarly to~the~situation on~Fig.~\ref{linkaxispic}, several ways of attachment are possible). Let us choose the~dimensions of~a~flag:~$2w_f = \delta - 2r - 4 = w + O(1)$, and~$2h_f = h$.

\begin{center}
\begin{figure}[htb]
\centering
\begin{subfigure}{\textwidth}
\centering
\includegraphics[scale=2]{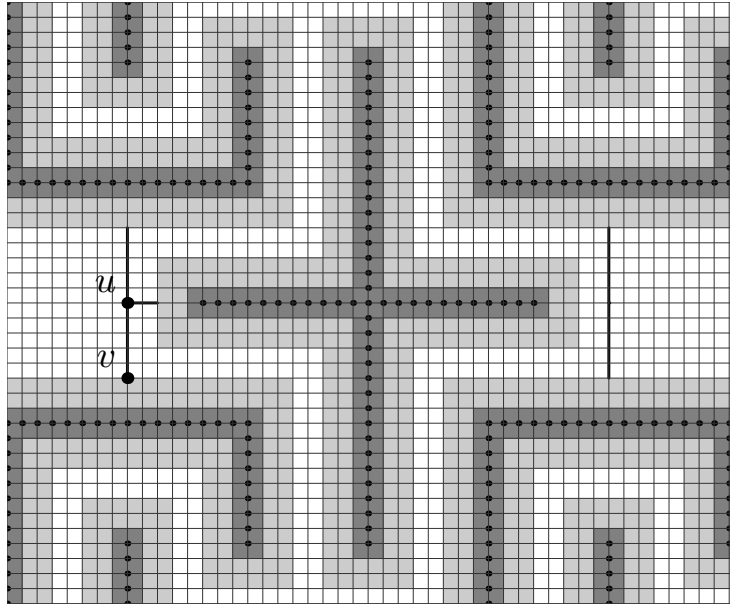}
\end{subfigure}
\captionsetup{justification=centering}
\caption{A~flag and its attachment to~an~auxiliary edge chain}
\label{flagspic1}
\end{figure}
\end{center}

The~length of~the~$uv$-path is~$O(1)$, and the flag is connected to~the~vertex~$u$ by a chain of two edges. We have shown above that the~shift between two adjacent links on~the~same level is~$(\delta + O(1), O(1))$, and the~shift between~consecutive links on~a~same chain is~$(O(1), h + O(1))$. A~reasoning similar to~the~one in~Section~\ref{links} can show that under restrictions \begin{align} w_f > c_7, h_f > c_8 \end{align}
bars of~the~cross will be aligned with~the~axes in~any embedding.

Let us ensure that two adjacent flags on~the~same level can not point towards each other. Let us enforce \begin{align} \label{crossver} 2h_f > l + 2r. \end{align} Now the~vertical bar can not fit into~the~vertical gap between two links on~the~same chain, hence it has to~be located in~the~space between the~chains. But for~this space to~accomodate two crosses simultaneously, the~gap must be at~least~$w_f$ wide; however, it is only~$\delta - w + O(1) = O(1)$ wide. From this, we obtain the restriction: \begin{align} \label{crosshor} w_f > c_9. \end{align}

Finally, if a~flag is pointed towards the~wall of~the~frame, then we must have~$\delta / 2 > w_f + O(1)$, which is equivalent to~$w / 2 < O(1)$, hence we obtain the~final~requirement \begin{align} \label{crosswall} w > c_{10}. \end{align}

\subsubsection{Choosing the parameters}

It suffices now to~choose suitable values for~parameters~$w$, $h$, $W$, $H$, $w_f$, $h_f$ to~satisfy the~restrictions~\eqref{linkshor}-\eqref{crosswall}, since we have established all necessary features of~the~construction required~to~implement the~correct logic engine behaviour. It can be verified that we can choose~$w = O(1)$,~$h = O(n)$. We now have that all vertices of~the~graph~$X$ fit inside a~rectangle of~dimensions $O(n) \times O(nm)$ in~the~basis~$\B$. Vertices of~$G$ are placed discretely in~$\B$, hence the~graph~$X$ has size~$O(n^2 m)$.

Thus, the~construction of~the~graph~$X$ is a~valid polynomial reduction from \textbf{LOGIC-ENGINE} to~injective embeddability in~$\R^1$, hence the~latter is NP-hard.

\subsubsection{Strict embeddability}

Up to this point, we have only considered non-strict injective embeddings of~$X$. However, note that if~$X$ is injectively embeddable, then it must be strictly injectively embeddable as well. Indeed, let~$\phi$ denote an~injective embedding of~$X$ reconstructed from~a~logic engine realization in~such a~way that all auxiliary chains are aligned with~corresponding axes (as~on~Fig.~\ref{schemepic}). It can be verified that in~$\phi$ the~gray regions of~different parts are at~distance at~least~2 apart, and~$\phi$ is locally induced in~all chain attachment points, thus~$\phi(X)$ is an~induced subgraph of~$\Gamma$, consequently,~$\phi$ is strict. Thus, the~reduction can be applied to strict injective embeddability just as well, and the~complexity result is naturally extended.

Finally, to~establish NP-hardness of~strict non-injective embeddability, we invoke Prop.~\ref{notinj}, and point out that no two vertices of $X$ have the~same neighbourhood by~construction of~$X$. This concludes the~proof of~the~(b) part of~Theorem~\ref{summary} in~the~$k = 2$~case.

\subsection{Construction for~the~$k > 2$ case} \label{z3}

\begin{center}
\begin{figure}[htb]
\centering
\begin{subfigure}{\textwidth}
\centering
\includegraphics{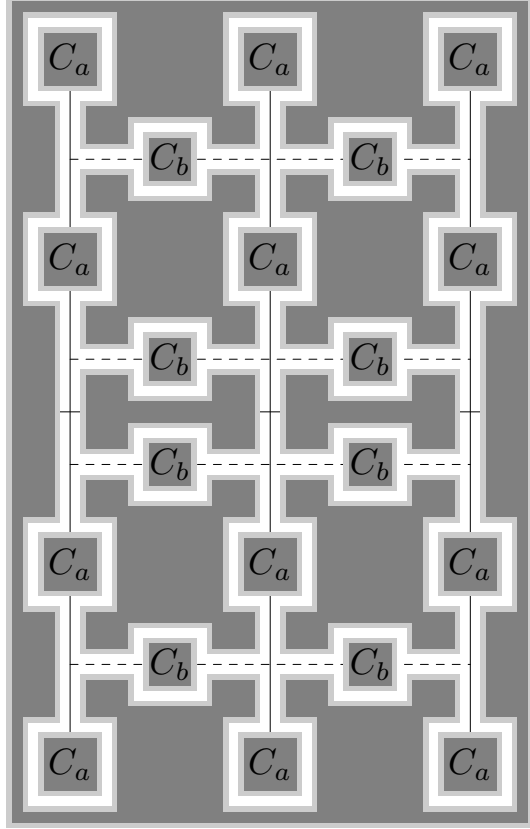}
\end{subfigure}
\captionsetup{justification=centering, margin=2cm}
\caption{Construction scheme for~the~$k > 2$ case. Dashed chains are introduced depending on whether a flag is present in~the~corresponding place of~the~logic engine.}
\label{schemekdpic}
\end{figure}
\end{center}

As~before, we choose~$r = \max(R(\A), \max_{||x|| \leq 1} \rho(x))$. For~an~integer~$a$, let~$C_a$ denote the~ball bundle with~centers in~integer points of~the~hypercube region~$S_a$.

The reduction of {\bf LOGIC-ENGINE} to $\A$-embeddability in $\R^1$ for the $k > 2$ case works as follows. Let the input logic engine contain $n$ rods, and $m$ flag levels on each side of the axle. Choose parameters $a$ and $b$. Construct the graph $X'$ as a union of $2nm$ copies of the graph $C_a$ that emulate ``chain links'', $2(n - 1)m$ copies of $C_b$ that emulate flags, and auxiliary edge chains according to the Fig. \ref{schemekdpic}. Here the horizontal direction corresponds to the axis $0b_1$, and the vertical direction to the axis $0b_2$, where $b_1$ and $b_2$ are elements of the basis $\B$ (all other elements of $\B$ are orthogonal to the displayed plane).

Choose $L$ as a large enough parallelepiped that encompasses $X'$. We choose $L$ so that the set $L'$ containing all integer points at least $r + 2$ away (with respect to $l_{\infty}$ norm) of gray regions of balls of $X'$ and auxiliary chains is connected in $\Gamma$. Construct a ``framework'' graph $Y$ as an induced union of balls with vertices in $L'$. The black regions of balls of $Y$ contain $L$ with removed neighbourhoods of all $C_a$ and $C_b$, and ``corridors'' of width $2r + 4$ around auxiliary chains. To obtain the graph $X$, attach all auxiliary chains along with the paths anchoring them to $Y$, and also attach the flags according to the input logic engine configuration.

As with the $k = 2$ case, the restriction of any embedding of $X$ to the vertices of the ball bundle $Y$ is an affine automorphism of $\Gamma$, hence we can assume that any embedding of $X$ acts trivially on $Y$. We claim that for $a, b \geq D$ and large enough $D$ we have that in any embedding of $X$ all images of $C_a$ and $C_b$ must be aligned with coordinate axes.

Consider a copy of $C_D$, and assume that in any embedding of $X$ its image lies completely within $L$. First, observe that for a large enough $D$ the image of the center of $C_D$ cannot lie in any ``corridor'' of width $2r + 4$. Indeed, the restriction of an embedding to vertices of $C_D$ is composed of a translation and a non-degenerate linear map $T: \R^k \to \R^k$ that is a unique extension of an element of $\Aut_0(\Gamma)$. For~a~particular~$T$, choose $D$ so large that $T(S_D)$ contains all points at most $r + 3$ away from the cube center. We now have that the black region of $T(C_D)$ cannot fit in any corridor by at least one dimension. FInally, observe that there are only finitely many options for $T$, hence we can choose $D$ that excludes placing $C_D$ in a corridor for each of the options.

Let the center of $C_D$ now lie an a cubic neighbourhood of size $D + 2r + 2$. Let $T$ denote the same linear map as in the previous paragraph. Observe that $T(\Conv \A) = \Conv \A$, hence $T$ is volume-preserving (i.e. has Jacobian equal to $\pm 1$). Consider the bounding box (that is, the least enclosing axes-aligned parallelepiped) of~the~set~$T(S_D)$ denoted as $P_D$. Let $(p_1(D), \ldots, p_k(D))$ denote the linear dimensions of $P_D$. If~$T(C_D) \neq C_D$, then $p_i(D) > D$ for a certain coordinate $i$. Moreover, the numbers $p_i(D)$ are linear~in~$D$.

The faces of the cubic neighbourhoods in our construction are allowed to have ``windows'', that is, openings of corridors of width $2r + 4$. Let us show that image of a large enough $C_D$ does not fit in a cubic neighbourhood with windows. Let $H_z$ denote the hyperplane $x_i = z$ (with the index $i$ chosen above). Suppose that the center of the image of $C_D$ has coordinate $i$ equal to $x_i$ relative to the center of the cubic neighbourhood. Then we must have both
\[\Vol_{k - 1}(T(S_D) \cap H_{D + r + 1 - x_i}) \leq (2r + 4)^{k - 1}, \Vol_{k - 1}(T(S_D) \cap H_{D + r + 1 + x_i}) \leq (2r + 4)^{k - 1},\] since otherwise the black region of $T(C_D)$ does not fit into the windows of $(k - 1)$-dimensional volume $(2r + 4)^{k - 1}$. The value $\min(\Vol_{k - 1}(T(S_D) \cap H_{D + r + 1 \pm x_i}))$ is attained for $x_i = 0$, so it suffices to show \[\Vol_{k - 1}(T(S_D) \cap H_{D + r + 1}) > (2r + 4)^{k - 1}\] for sufficiently large $D$.

Let $q$ be the point of $T(S_1)$ with the largest coordinate $i$, and $t = p_i(1) > 1$ be the value of this coordinate. The $(k-1)$-dimensional volume $\Vol_{k - 1}(T(S_1) \cap H_1) = z_1$ must be positive. We have that $\Vol_{k - 1}(T(S_D) \cap H_D) = z_1 D^{k - 1}$. Further, by convexity the set $T(S_D) \cap H_{D + r + 1}$ must contain the homothetical image of $T(S_D) \cap H_D$ with the homothetic center $Dq$ and coefficient $(Dt - (D + r + 1)) / (Dt - D)$, hence
\[\Vol_{k - 1}(T(S_D) \cap H_{D + r + 1}) \geq z_1 D^{k - 1} \cdot \left(\frac{Dt - (D + r + 1)}{Dt - D}\right)^{k - 1} = z_1 D^{k - 1}\left(1 - \frac{r + 1}{D(t - 1)}\right)^{k - 1}\]

Clearly, the right-hand side can be made larger than $(2r + 4)^{k - 1} = O(1)$ by choosing a large enough $D$. It follows that $T(C_D) = C_D$ in any embedding for a large enough $D$.

Let us now choose suitable values for $a$ and $b$. Let $a, b \geq D$ and $a > b + 2r + 2$, then in any embedding images of all copies of $C_a$ and $C_b$ must be axes-aligned, and also copies of~$C_a$ cannot fit into neighbourhoods of $C_b$. Let us further impose $2(2a)^k > (2a + 2r + 2)^k$, $2(2b)^k > (2b + 2r + 2)^ k$, $(2a)^k + (2b)^k > (2a + 2r + 2)^k$. Under these restrictions no two cube copies cannot lie within the same neighbourhood due to natural volume inequalities. Finally, due to distance limitations imposed by the auxiliary edge chains, the chain links will be positioned vertically, and each flag can only go to the slot nearest to the anchoring point. It follows that each copy of $C_a$ will lie in a neighbourhood of size $a + 2r + 2$, hence the copies of $C_b$ will lie in neighbouhoods of size $b + 2r + 2$. Consequently, we can restore a logic engine realization from any embedding of $X$. Finally, we are free to choose any suitable directions for link chains and flags, hence a logic engine realization can be turned to an embedding of $X$. In this way, the two problems are seen to be equivalent.

Since the chosen $a$ and $b$ are independent on the input, we have $a = O(1)$, $b = O(1)$, and the graph $X$ is enclosed in a parallelepiped with dimensions $O(n) \times O(m) \times O(1) \times \ldots \times O(1)$, and its size is $O(nm)$. Thus the reduction from {\bf LOGIC-ENGINE} is polynomial, and the (b) case of Theorem \ref{summary} is established for injective embeddings.

To conclude the proof of Theorem \ref{summary} we consider strict non-injective $\A$-embeddings for $k > 1$. It can be verified explicitly that the reduction graphs constructed in Sections~\ref{z2} and~\ref{z3} do not have vertices with equal neighbourhoods. 
Consequently,~Proposition \ref{notinj} implies that each of these graphs is strictly $\A$-embeddable if and only if it is strictly injectively $\A$-embeddable, hence the same constructions work for reducing {\bf LOGIC-ENGINE} to the strict embeddability problem. Theorem \ref{summary} is now proven completely.

\printbibliography

\end{document}